\documentclass[final]{article}
\usepackage{graphicx} \usepackage{microtype} \usepackage{xspace}
\usepackage{mathtools} \usepackage{stmaryrd} \usepackage{amssymb}
\usepackage{bbm} \usepackage{enumerate} \usepackage{standalone}
\usepackage[ruled,linesnumbered]{algorithm2e}
\usepackage{tikz} \usetikzlibrary{arrows, automata, positioning, chains,
  patterns, decorations.pathreplacing, calc}
\usepackage{marvosym}
\usepackage[bookmarks,unicode,hidelinks,breaklinks]{hyperref}

\usepackage{amsthm}
\newtheorem{theorem}{Theorem}

\newtheorem{proposition}[theorem]{Proposition}

\newtheorem{lemma}[theorem]{Lemma}

\newtheorem{corollary}[theorem]{Corollary}

\theoremstyle{definition}
\newtheorem{definition}{Definition}

\theoremstyle{remark}
\newtheorem{remark}{Remark}

\newtheorem{example}{Example}

\newcommand{\myparagraph}[1]{\paragraph{{#1}}}

\usepackage{verbatim} \usepackage{slashbox}

\newcommand{\N}{\mathbb{N}} \newcommand{\Z}{\mathbb{Z}}
\DeclarePairedDelimiter\size{\lvert}{\rvert}
\DeclarePairedDelimiter\length{\lVert}{\rVert}
\DeclarePairedDelimiter\rel{\llbracket}{\rrbracket}
 \newcommand{\dom}{\mathrm{dom}}

\newcommand{\inp}{\mathbbm{i}} \newcommand{\outp}{\mathbbm{o}}
\newcommand{\dz}{\textsf{d}_0} \newcommand{\data}{\mathcal{D}}

\newcommand{\NLogSpace}{\textsc{NLogSpace}\xspace}
\newcommand{\PTime}{\textsc{PTime}\xspace}
\newcommand{\PSpace}{\textsc{PSpace}\xspace}
\newcommand{\nra}{\textnormal{\textsf{NRA}}\xspace}

\newcommand{\nrt}{\textnormal{\textsf{NRT}}\xspace}
\newcommand{\drt}{\textnormal{\textsf{DRT}}\xspace}
\newcommand{\syn}{\textnormal{\textsf{syn}}\xspace}

\newcommand{\nft}{\textnormal{\textsf{NFT}}\xspace} \newcommand{\tst}{\phi}

\newcommand{\Tst}{\textnormal{\textsf{Tst}}}
\newcommand{\asgn}{\textnormal{\textsf{asgn}}}
\newcommand{\projin}{\textnormal{\textsf{in}}}
\newcommand{\projout}{\textnormal{\textsf{out}}}
\newcommand{\states}{\textnormal{\textsf{st}}}
\newcommand{\trace}{\textnormal{\textsf{tr}}}
\newcommand{\lab}{\textsf{lab}} \newcommand{\dt}{\textsf{dt}}

\newcommand{\mismatch}{\textnormal{\textsf{mismatch}}}

\usepackage{calc} \newcommand\myxrightarrow[2][]{
  \xrightarrow[{\raisebox{1.25ex-\heightof{$\scriptstyle#1$}}[0pt]{$\scriptstyle#1$}}]{#2}%
}
\makeatletter \providecommand*{\shuffle}{%
  \mathbin{\mathpalette\shuffle@{}}%
} \newcommand*{\shuffle@}[2]{%
  \sbox0{$#1\vcenter{}$}%
  \kern .15\ht0 
  \rlap{\vrule height .25\ht0 depth 0pt width 2.5\ht0}%
  \raise.1\ht0\hbox to 2.5\ht0{%
    \vrule height 1.75\ht0 depth -.1\ht0 width .17\ht0 %
    \hfill \vrule height 1.75\ht0 depth -.1\ht0 width .17\ht0 %
    \hfill \vrule height 1.75\ht0 depth -.1\ht0 width .17\ht0 %
  }%
  \kern .15\ht0 
} \makeatother

\makeatletter
\def\orcidID#1{\smash{\href{http://orcid.org/#1}{\protect\raisebox{-1.25pt}{\protect\includegraphics{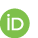}}}}}
\makeatother

\usepackage{authblk}

\begin{document}
\title{On Computability of Data Word Functions Defined by
  Transducers} 
%
%

\author[1,2]{L\'eo Exibard\thanks{Corresponding author: \texttt{leo.exibard@ulb.ac.be}.
    Funded by a FRIA fellowship from the F.R.S.-FNRS.}\orcidID{0000-0003-0318-1217}} 
\author[1]{Emmanuel Filiot\thanks{Research associate of F.R.S.-FNRS. He is
    supported by the ARC Project Transform F\'ed\'eration Wallonie-Bruxelles and
    the FNRS CDR J013116F and MIS F451019F projects.}}
\author[2]{Pierre-Alain Reynier\thanks{Partly funded by the DeLTA project (ANR-16-CE40-0007).}}
%
%
\affil[1]{Universit\'e Libre de Bruxelles, Brussels, Belgium}
\affil[2]{Aix-Marseille Universit\'e, Marseille, France}
\date{}
\maketitle 
\begin{abstract}
  In this paper, we investigate the problem of synthesizing computable functions
  of infinite words over an infinite alphabet (data $\omega$-words). The notion
  of computability is defined through Turing machines with infinite inputs which
  can produce the corresponding infinite outputs in the limit. We use
  non-deterministic transducers equipped with registers, an extension of
  register automata with outputs, to specify functions. Such transducers may not
  define functions but more generally relations of data $\omega$-words, and we
  show that it is \PSpace-complete to test whether a given transducer defines a
  function. Then, given a function defined by some register transducer, we show
  that it is decidable (and again, \PSpace-c) whether such function is
  computable. As for the known finite alphabet case, we show that computability
  and continuity coincide for functions defined by register transducers, and
  show how to decide continuity. We also define a subclass for which those
  problems are \PTime.

  \bigskip

\end{abstract}
\section{Introduction}


\myparagraph{Context} Program synthesis aims at deriving, in an automatic way, a
program that fulfils a given specification. Such setting is very appealing when
for instance the specification describes, in some abstract formalism (an
automaton or ideally a logic), important properties that the program must
satisfy. The synthesised program is then \emph{correct-by-construction} with
regards to those properties. It is particularly important and desirable for the
design of safety-critical systems with hard dependability constraints, which are
notoriously hard to design correctly.

Program synthesis is hard to realise for general-purpose programming languages
but important progress has been made recently in the automatic synthesis of
\emph{reactive systems}. In this context, the system continuously receives input
signals to which it must react by producing output signals. Such systems are not
assumed to terminate and their executions are usually modelled as infinite words
over the alphabets of input and output signals. A specification is thus a set of
pairs (\textsf{in},\textsf{out}), where \textsf{in} and \textsf{out} are
infinite words, such that \textsf{out} is a legitimate output for \textsf{in}.
Most methods for reactive system synthesis only work for \emph{synchronous}
systems over \emph{finite} sets of input and output signals $\Sigma$ and
$\Gamma$. In this synchronous setting, input and output signals alternate, and
thus \emph{implementations} of such a specification are defined by means of
\emph{synchronous} transducers, which are B\"uchi automata with transitions of
the form $(q,\sigma,\gamma,q')$, expressing that in state $q$, when getting
input $\sigma\in\Sigma$, output $\gamma\in\Gamma$ is produced and the machine
moves to state $q'$. We aim at building \emph{deterministic} implementations, in
the sense that the output $\gamma$ and state $q'$ uniquely depend on $q$ and
$\sigma$.
The realisability problem of specifications given as synchronous
non-deterministic transducers, by implementations defined by synchronous
deterministic transducers is known to be decidable~\cite{BuLa69,PnuRos:89}. In
this paper, we are interested in the \emph{asynchronous} setting, in which
transducers can produce none or several outputs at once every time some input is
read, i.e., transitions are of the form $(q,\sigma,w,q')$ where $w\in \Gamma^*$.
However, such generalisation makes the realisability problem
undecidable~\cite{CarayolL14,DBLP:conf/icalp/FiliotJLW16}.

\myparagraph{Synthesis of Transducers with Registers} In the above setting,
the set of signals is considered to be finite. This assumption is not
realistic in general, as signals may come with unbounded information (e.g.
process ids) that we call here \emph{data}. To address this limitation, recent
works have considered the synthesis of reactive systems processing \emph{data
  words}~\cite{DBLP:conf/atva/KhalimovMB18,Ehlers:2014:SI:2961203.2961226,DBLP:conf/concur/KhalimovK19,DBLP:conf/concur/ExibardFR19}.
Data words are infinite words over an alphabet $\Sigma\times \data$, where
$\Sigma$ is a finite set and $\data$ is a possibly infinite countable set. To
handle data words, just as automata have been extended to \emph{register
  automata}, transducers have been extended to \emph{register transducers}. Such
transducers are equipped with a finite set of registers in which they can store
data and with which they can compare data for equality or inequality. While the
realisability problem of specifications given as synchronous non-deterministic
register transducers ($\nrt_\syn$) by implementation defined by synchronous
deterministic register transducers ($\drt_\syn$) is undecidable, decidability is
recovered for specifications defined by universal register transducers and by
giving as input the number of registers the implementation must
have~\cite{DBLP:conf/concur/ExibardFR19,DBLP:conf/atva/KhalimovMB18}.

\myparagraph{Computable Implementations} In the previously mentioned works, both
for finite or infinite alphabets, implementations are considered to be
deterministic transducers. Such an implementation is guaranteed to use only a
constant amount of memory (assuming data have size $O(1)$). While it makes sense
with regards to memory-efficiency, some problems turn out to be undecidable, as
already mentioned: realisability of $\nrt_\syn$ specifications by $\drt_\syn$,
or, in the finite alphabet setting, when both the specification and
implementation are asynchronous. In this paper, we propose to study computable
implementations, in the sense of (partial) functions $f$ of data $\omega$-words
computable by some Turing machine $M$ that has an infinite input $x\in \dom(f)$,
and must produce longer and longer prefixes of the output $f(x)$ as it reads
longer and longer prefixes of the input $x$. Therefore, such a machine produces
the output $f(x)$ in the limit. We denote by $\textsf{TM}$ (for Turing machine)
the class of Turing machines computing functions in this sense. As an example,
consider the function $f$ that takes as input any data $\omega$-word $u =
(\sigma_1,d_1)(\sigma_2,d_2)\dots$ and outputs $(\sigma_1,d_1)^\omega$ if $d_1$
occurs at least twice in $u$, and otherwise outputs $u$. This function is not
computable, as an hypothetic machine could not output anything as long as $d_1$
is not met a second time. However, the following function $g$ is computable. It
is defined only on words $(\sigma_1,d_1)(\sigma_2,d_2)\dots$ such that
$\sigma_1\sigma_2\dots\in ((a+b)c^*)^\omega$, and transforms any
$(\sigma_i,d_i)$ by $(\sigma_i, d_1)$ if the next symbol in $\{a,b\}$ is an $a$,
otherwise it keeps $(\sigma_i,d_i)$ unchanged. To compute it, a \textsf{TM}
would need to store $d_1$, and then wait until the next symbol in $\{a,b\}$ is
met before outputting something. Since the finite input labels are necessarily
in $((a+b)c^*)^\omega$, this machine will produce the whole output in the limit.
Note that $g$ cannot be defined by any deterministic register transducer, as it
needs unbounded memory to be implemented.

However, already in the finite alphabet setting, the problem of deciding if a
specification given as some non-deterministic synchronous transducer is
realisable by some computable function is open. The particular case of
realisability by computable functions of universal domain (the set of all
$\omega$-words) is known to be
decidable~\cite{DBLP:journals/corr/abs-1209-0800}. In the asynchronous setting,
the undecidability proof of~\cite{CarayolL14} can be easily adapted to show the
undecidability of realisability of specifications given by non-deterministic
(asynchronous) transducers by computable functions.

\myparagraph{Functional Specifications} As said before, a specification is in
general a relation from inputs to outputs. If this relation is a function, we
call it functional. Due to the negative results just mentioned about the
synthesis of computable functions from non-functional specifications, we instead
here focus on the case of functional specifications and address the following
general question: given the specification of a function of data $\omega$-words,
is this function ``implementable'', where we define ``implementable'' as ``being
computable by some Turing machine''. Moreover, if it is implementable, then we
want a procedure to automatically generate an algorithm that computes it. This
raises another important question: how to decide whether a specification is
functional ? We investigate these questions for asynchronous register
transducers, here called register transducers. This asynchrony allows for much
more expressive power, but is a source of technical challenge.

\myparagraph{Contributions} In this paper, we solve the questions mentioned
before for the class of (asynchronous) non-deterministic register transducers
(\nrt). We also give fundamental results on this class. In particular, we prove
that:
\begin{enumerate}
\item deciding whether an \nrt defines a function is \PSpace-complete,
\item deciding whether two functions defined by \nrt are equal on the
  intersection of their domains is \PSpace-complete,
\item the class of functions defined by \nrt is effectively closed under
  composition,
\item computability and continuity are equivalent notions for functions defined
  by \nrt, where continuity is defined using the classical Cantor distance,
\item deciding whether a function given as an \nrt is computable is \PSpace-c,
\item those problems are in \textsc{PTime} for a subclass of \nrt, called
  test-free \nrt.
\end{enumerate}

Finally, we also mention that considering the class of deterministic register
transducers (\drt for short) instead of computable functions as a yardstick for
the notion of being ``implementable'' for a function would yield undecidability.
Indeed, given a function defined by some \nrt, it is in general undecidable to
check whether this function is realisable by some \drt, by a simple reduction
from the universality problem of non-deterministic register
automata~\cite{Neven:2004:FSM:1013560.1013562}.


\myparagraph{Related Work} The notion of continuity with regards to Cantor
distance is not new, and for rational functions over finite alphabets, it was
already known to be decidable~\cite{DBLP:journals/tcs/Prieur02}. Its connection
with computability for functions of $\omega$-words over a finite alphabet has
recently been investigated in~\cite{DBLP:journals/corr/abs-1906-04199} for
one-way and two-way transducers. Our results lift some of theirs to the setting
of data words. The model of test-free \nrt can be seen as a one-way
non-deterministic version of a model of two-way transducers considered
in~\cite{DBLP:conf/fossacs/Durand-Gasselin16}.

\myparagraph{Outline} In \autoref{sec:prelims}, we define register
transducers and automata. In \autoref{sec:functions}, we study functions defined
by \nrt: we show that their functionality problem is decidable and that they
are closed under composition. In
\autoref{sec:computability-continuity}, we connect computability and continuity
for those functions, and prove that these notions are decidable.
Finally, we study the test-free restriction in \autoref{sec:testFree}.

\section{Data Words and Register Transducers}
\label{sec:prelims}

For a (possibly infinite) set $S$, we denote by $S^*$ (resp. $S^\omega$) the set
of finite (resp. infinite) words over this alphabet, and we let $S^\infty = S^*
\cup S^\omega$. For a word $u = u_1 \dots u_n$, we denote $\length{u} = n$ its length,
and, by convention, for $u \in S^\omega, \length{u} = \infty$. The empty word is denoted $\varepsilon$. For $1\leq i\leq
j \leq \length{u}$, we let $u[i{:}j] = u_i u_{i+1} \dots u_j$ and $u[i] =
u[i{:}i]$ the $i$th letter of $u$. For $u,v \in S^\infty$, we say that $u$ is a
prefix of $v$, written $u \preceq v$, if there exists $w \in S^\infty$ such that
$v = uw$. In this case, we define $u^{-1}v = w$. For $u,v \in S^\infty$, we say
that $u$ and $v$ \emph{mismatch}, written $\mismatch(u,v)$, when there exists a
position $i$ such that $1 \leq i \leq \length{u}$, $1 \leq i \leq \length{v}$
and $u[i] \neq v[i]$. Finally, for $u,v \in S^\infty$, we denote by $u \wedge v$
their longest common prefix, i.e. the longest word $w \in S^\infty$ such that $w
\preceq u$ and $w \preceq v$.

\myparagraph{Data Words} In the whole paper, $\Sigma$ and $\Gamma$ are two
finite alphabets and $\mathcal{D}$ is a countably infinite set of elements
called \emph{data}. We will use letter $\sigma$ (resp. $\gamma$, $d$) to denote
elements of $\Sigma$ (resp. $\Gamma$, $\data$). We also distinguish an
(arbitrary) data value $\dz\in\data$. Given a set $R$, let $\tau_0^R$ be the
constant function defined by $\tau_0^R(r)=\dz$ for all $r\in R$. Given a finite
alphabet $A$, a \emph{labelled data} is a pair $x = (a,d)\in A\times \data$,
where $a$ is the \emph{label} and $d$ the \emph{data}. We define the projections
$\lab(x) = a$ and $\dt(x) = d$. A \emph{data word} over $A$ and $\mathcal{D}$ is
an infinite sequence of labelled data, i.e. a word $w\in (A \times
\mathcal{D})^\omega$. We extend the projections $\lab$ and $\dt$ to data words
naturally, i.e. $\lab(w)\in A^\omega$ and $\dt(w) \in \data^\omega$. A
\emph{data word language} is a subset $L\subseteq (A \times \data)^\omega$. Note
that in this paper, data words are infinite, otherwise they are called
\emph{finite data words}.

\subsection{Register Transducers}
Register transducers are transducers recognising data word relations. They are
an extension of finite transducers to data word relations, in the same way
register automata~\cite{Kaminski:1994:FA:194527.194534} are an extension of
finite automata to data word languages. Here, we define them over infinite data words
with a B\"uchi acceptance condition, and allow multiple registers to contain the
same data, with a syntax close to~\cite{DBLP:journals/jcss/LibkinTV15}.
The current data can be 
compared for equality with the register contents via tests, which
are symbolic and defined via Boolean formulas of the following
form. Given $R$ a set of registers, a \emph{test} is a formula $\phi$
satisfying the following syntax:
$$
\phi\ ::=\ \top\mid \bot\mid r^=\mid r^{\neq} \mid \phi\wedge \phi \mid \phi\vee
\phi\mid \neg \phi
$$
where $r\in R$. Given a valuation $\tau : R\rightarrow \data$, a
test $\phi$ and a data $d$, we denote by $\tau,d\models \phi$ the
satisfiability of $\phi$ by $d$ in valuation $\tau$, defined as $\tau,d\models r^=$
if $\tau(r) = d$ and $\tau,d\models r^{\neq}$ if $\tau(r)\neq d$. The
Boolean combinators behave as usual. We denote by $\Tst_R$ the set of
(symbolic) tests over $R$. 
\begin{definition}
  A non-deterministic register transducer (\nrt) is a tuple $T = (Q,
  R, i_0, F, \Delta)$, where $Q$ is a finite set of \emph{states}, $i_0 \in Q$
  is the \emph{initial} state, $F \subseteq Q$ is the set of \emph{accepting}
  states, $R$ is a finite set of \emph{registers} and $\Delta \subseteq Q \times
  \Sigma \times \Tst_R \times 2^R \times (\Gamma \times R)^* \times Q$ is a finite
  set of \emph{transitions}. We write $q \myxrightarrow[T]{\sigma, \tst \mid
    \asgn, o} q'$ for $(q,\sigma,\tst,\asgn,o,q') \in \Delta$ (we omit $T$ when
  clear from the context).
\end{definition}
The semantics of a register transducer is given by a
labelled transition system: we define $L_T = (C, \Lambda, \rightarrow)$, where
$C = Q \times (R \rightarrow \data)$ is the set of configurations, $\Lambda =
(\Sigma \times \data) \times (\Gamma \times \data)^*$ is the set of labels, and
we have, for all $(q,\tau), (q',\tau') \in C$ and for all $(l,w) \in \Lambda$,
that $(q,\tau) \myxrightarrow{(l,w)} (q',\tau')$ whenever there exists a
transition $q \myxrightarrow[T]{\sigma, \tst \mid \asgn, o} q'$
such that, by writing $l = (\sigma',d)$ and $w = (\gamma'_1,d_1) \dots
(\gamma'_n,d_n)$:
\begin{itemize}
\item (Matching labels) $\sigma = \sigma'$
\item (Compatibility) $d$ satisfies the test $\tst \in \Tst_R$, i.e. 
  $\tau,d \models \tst$.
\item (Update) $\tau'$ is the successor register configuration of $\tau$ with
  regards to $d$ and $\asgn$: $\tau'(r)=d$ if $r \in \asgn$, and $\tau'(r) =
  \tau(r)$ otherwise
\item(Output) By writing $o = (\gamma_1, r_1) \dots (\gamma_m, r_m)$, we have
  that $m=n$ and for all $1 \leq i \leq n$, $\gamma_i = \gamma'_i$ and $d_i =
  \tau'(r_i)$.
\end{itemize}

Then, a \emph{run} of $T$ is an infinite sequence of configurations and
transitions $\rho = (q_0, \tau_0) \myxrightarrow[L_T]{(u_1,v_1)} (q_1,\tau_1)
\myxrightarrow[L_T]{(u_2,v_2)} \cdots$. Its input is $\projin(\rho) = u_1 u_2
\dots$, its output is $\projout(\rho) = v_1 \cdot v_2 \dots$. We also define its
sequence of states $\states(\rho) = q_0 q_1 \dots$, and its \emph{trace}
$\trace(\rho) = u_1 \cdot v_1 \cdot u_2 \cdot v_2 \dots$. Such run is
\emph{initial} if $(q_0,\tau_0) = (i_0, \tau_0^R)$. It is \emph{final} if it
satisfies the B\"uchi condition, i.e. $\mathrm{inf}(\states) \cap F \neq
\varnothing$, where $\mathrm{inf}(\states) = \{q \in Q \mid q = q_i$ for
infinitely many $i\}$.
Finally, it is \emph{accepting} if it is both initial and final. We then write
$(q_0,\tau_0) \myxrightarrow[T]{u \mid v}$ to express that there is a final run
$\rho$ of $T$ starting from $(q_0,\tau_0)$ such that $\projin(\rho) = u$ and
$\projout(\rho) = v$. In the whole paper, and unless stated otherwise, we always
assume that the output of an accepting run is infinite ($v \in (\Gamma \times
\data)^\omega$), which can be ensured by a B\"uchi condition.

A \emph{partial run} is a finite prefix of a run. The notions of input, output
and states are extended by taking the corresponding prefixes. We then write
$(q_0,\tau_0) \myxrightarrow[T]{u \mid v} (q_n,\tau_n)$ to express that there is
a partial run $\rho$ of $T$ starting from configuration $(q_0,\tau_0)$ and
ending in configuration $(q_n,\tau_n)$ such that $\projin(\rho) = u$ and
$\projout(\rho) = v$.

Finally, the relation represented by a transducer $T$ is:
\begin{align*}
  \rel{T} = \bigl\{(u,v) \in (\Sigma \times \data)^\omega \times (\Gamma \times \data)^\omega \mid\ &\text{there exists an accepting run $\rho$ of $T$} \\
                                                                                                    & \text{such that } \projin(\rho) = u \text{ and } \projout(\rho) = v\bigr\}
\end{align*}

\begin{example}\label{ex:Trename}
  As an example, consider the register transducer
  $T_{\mathsf{rename}}$ depicted in \autoref{fig:Trename}. It
  realises the following transformation: consider a setting in which
  we deal with logs of communications between a set of clients.  Such
  a log is an infinite sequence of pairs consisting of a tag, chosen
  in some finite alphabet $\Sigma$, and the identifier of the client
  delivering this tag, chosen in some infinite set of data values. The
  transformation should modify the log as follows: for a given client
  that needs to be modified, each of its messages should now be
  associated with some new identifier. The transformation has to verify
  that this new identifier is indeed free, \emph{i.e.} never used in
  the log.  Before treating the log, the transformation receives as
  input the id of the client that needs to be modified (associated
  with the tag $\mathsf{del}$), and then a sequence of identifiers (associated
  with the tag $\mathsf{ch}$), ending with $\#$. The
  transducer is non-deterministic as it has to guess which of these
  identifiers it can choose to replace the one of the client. In
  particular, observe that it may associate multiple output words to
  a same input if two such free identifiers exist.
\end{example}
\begin{figure}[htb]
  \centering
\begin{tikzpicture}[->,>=stealth',auto,node distance=2cm,thick,scale=0.9,every node/.style={scale=0.85}]
\tikzstyle{every state}=[text=black, font=\scriptsize, fill=yellow!30,minimum size=7.5mm]

\tikzstyle{input char}=[text=Red3]
\tikzstyle{output char}=[text=Green4]

\node[state, initial, initial text={}]    (i) {1};
\node[state, right=of i]                  (p) {2};
\node[state, right=of p]                  (q) {3};
\node[state, right=of q, accepting]       (r) {4};

\path (i) edge node {$\mathsf{del},\top\mid r_1,\epsilon $} (p); \path (p)
 edge[loop above] node {$\mathsf{ch},\top \mid \varnothing,\epsilon $} (p);
 \path (p) edge node {$\mathsf{ch},r_1^{\neq} \mid r_2,\epsilon $} (q); \path
 (q) edge[loop above] node {$\mathsf{ch},\top \mid \varnothing,\epsilon $} (q);
 \path (q) edge node {$\#,\top\mid \varnothing,\epsilon $} (r); \path (r)
 edge[loop above] node {$\sigma,r_1^= \mid \varnothing,(\sigma,r_2) $} (r);
 \path (r) edge[loop below] node {$\sigma,r_1^{\neq} \wedge r_2^{\neq} \mid
   r_0,(\sigma, r_0) $} (r);
\end{tikzpicture}
  \caption{A register transducer $T_\mathsf{rename}$. It has three
    registers $r_1$, $r_2$ and $r_0$ and four states.
    $\sigma$ denotes any letter in $\Sigma$, $r_1$ stores the id of
    $\mathsf{del}$ and $r_2$ the chosen id of $\mathsf{ch}$, while $r_0$ is used
    to output the last data value read as input. As we only assign data to
    single registers, we write $r_i$ for the singleton assignment set $\{r_i\}$.}
  \label{fig:Trename}
\end{figure}
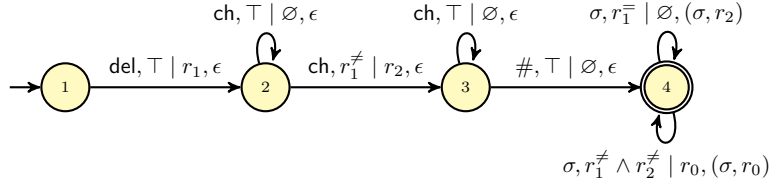

\myparagraph{Finite Transducers} Since we reduce the decision of continuity and
functionality of \nrt to the one of finite transducers, let us introduce them: a
finite transducer (\nft for short) is an \nrt with 0 registers (i.e. $R =
\varnothing$). Thus, its transition relation can be represented as $\Delta
\subseteq Q \times \Sigma \times \Gamma^* \times Q$.
A direct extension of the construction of~\cite[Proposition
1]{Kaminski:1994:FA:194527.194534}
allows to show that:
\begin{proposition}
  \label{prop:regularRestrictionTrd}
  Let $T$ be an \nrt with $k$ registers, and let $X \subset_f \data$ be a finite
  subset of data. Then, $\rel{T} \cap (\Sigma \times X)^\omega \times (\Gamma
  \times X)^\omega$ is recognised by an \nft of exponential size, more precisely
  with $O(\size{Q} \times \size{X}^{\size{R}})$ states.
\end{proposition}

\subsection{Technical Properties of Register Automata}
Although automata are simpler machines than transducers, we only use them
as tools in our proofs, which is why we define them from transducers, and not
the other way around. A non-deterministic register automaton, denoted \nra, is a
transducer without outputs: its transition relation is $\Delta \subseteq Q
\times \Sigma \times \Tst_R \times 2^R \times \{\varepsilon\} \times Q$ (simply
represented as $\Delta \subseteq Q \times \Sigma \times \Tst_R \times 2^R \times
Q$). The semantics are the same, except that now we lift the condition that the
output $v$ is infinite since there is no output. For $A$ an \nra, we denote
$L(A) = \{u \in (\Sigma \times \data)^\omega \mid\ $there exists an accepting
run $\rho$ of $A$ over $u\}$. Necessarily the output of an accepting run is
$\varepsilon$. In this section, we establish technical properties about \nra.

\autoref{prop:indistinguishability}, the so-called
``indistinguishability property'', was shown in the seminal paper by Kaminski
and Francez~\cite[Proposition 1]{Kaminski:1994:FA:194527.194534}. Their model
differs in that they do not allow distinct registers to contain the same data,
and in the corresponding test syntax, but their result easily carries to our
setting.
It states that if an \nra accepts a data word, then such data word can be
relabelled with data from any set containing $\dz$ and with at least $k+1$
elements. Indeed, at any point of time, the automaton can
only store at most $k$ data in its registers, so its notion of ``freshness'' is
a local one, and forgotten data can thus be reused as fresh ones. Moreover, as
the automaton only tests data for equality, their actual value does not matter,
except for $\dz$ which is initially contained in the registers.

Such ``small-witness'' property is fundamental to \nra, and will be paramount in
establishing decidability of functionality (\autoref{sec:functions}) and of
computability (\autoref{sec:computability-continuity}). We use it jointly with
\autoref{lem:otimes}, which states that the interleaving of the traces of runs
of an \nrt can be recognised with an \nra, and \autoref{lem:mismatch}, which
expresses that an \nra can check whether interleaved words coincide on some
bounded prefix, and/or mismatch before some given position.
\begin{proposition}[\cite{Kaminski:1994:FA:194527.194534}]
  \label{prop:indistinguishability}
  Let $A$ be an \nra with $k$ registers. If $L(A) \neq \varnothing$, then, for
  any $X \subseteq \data$ of size $\size{X} \geq k+1$ such that $\dz \in X$,
  $L(A) \cap (\Sigma \times X)^\omega \neq \varnothing$.
\end{proposition}

The runs of a register transducer $T$ can be flattened to their traces, so as to be
recognised by an \nra. Those traces can then be interleaved, in order to be
compared. The proofs of the following properties are quite straightforward, but
are given for the sake of completeness.

  Let $\rho_1 = (q_0,\tau_0) \myxrightarrow[L_T]{(u_1, u'_1)} (q_1,\tau_1)
  \dots$ and $\rho_2 = (p_0,\mu_0) \myxrightarrow[L_T]{(v_1, v'_1)} (p_1,\mu_1)
  \dots$ be two runs of a transducer $T$. Then, we define their
  \emph{interleaving} $\rho_1 \otimes \rho_2 = u_1 \cdot u'_1 \cdot v_1 \cdot
  v'_1 \cdot u_2 \cdot u'_2 \cdot v_2 \cdot v'_2 \dots$ and $L_\otimes(T) = \{\rho_1
  \otimes \rho_2 \mid \rho_1 \text{ and } \rho_2 \text{ are accepting runs
    of } T\}$. 
\begin{lemma}
  \label{lem:otimes}
  If $T$ has $k$ registers, then $L_\otimes(T)$ is recognised by an \nra with
  $2k$ registers.
\end{lemma}
\begin{proof}
  We construct an automaton $A$ with $2k$ registers which, on an input data word
  $u\in ((\Sigma\times \data)(\Gamma\times \data)^*)^\omega$ guesses, using
  non-determinism, a decomposition $u = u_1 \cdot u'_1 \cdot v_1 \cdot v'_1
  \cdot u_2 \cdot u'_2 \cdot v_2 \cdot v'_2 \dots$ as defined before. In
  particular for all $i$, $|u_i| = |v_i| = 1$. The automaton $A$ also maintains
  two copies of $T$ (that is why $A$ needs $2k$ registers). The first one makes
  sure that $u_1u'_1u_2u'_2\dots = \trace(\rho_1)$ for some run $\rho_1$, the
  other one checks that $v_1v'_1v_2v'_2\dots = \trace(\rho_2)$ for some run
  $\rho_2$. To check those two properties, $A$ simulates $T$ on both
  decompositions, using the two copies of $T$. For the input part (when reading
  $u_1$ or $v_1$), $A$ simply makes the same tests as $T$ does. For the output
  part, if $A$ simulates a transition of $T$ with output $o =
  (\sigma_1,r_1)\dots (\sigma_n,r_n)$ on some output $v_i$, it makes $n$
  transitions which successively check that $r_j$ is equal to the current data,
  for all $1\leq j\leq n$.
\end{proof}

\begin{lemma}
  \label{lem:mismatch}
  Let $i,j \in \N \cup \{\infty\}$.
  We define $M_j^i = \{u_1 u'_1 v_1 v'_1 \dots \mid \forall k \geq 1, u_k, v_k
  \in (\Sigma \times \data),
  u'_k,v'_k \in (\Gamma \times \data)^*,
  \forall 1 \leq k \leq j, v_k = u_k \text{ and } \length{u'_1 \cdot u'_2 \dots
    \wedge v'_1 \cdot v'_2 \dots} \leq i \}$. Then, $M_j^i$ is recognisable by
  an \nra with $2$ registers and with $1$ register if $i = \infty$.
\end{lemma}
\begin{proof}
  Checking that $\forall k \geq 1, u_k,v_k \in (\Sigma \times \data),
  u'_k,v'_k \in (\Gamma \times \data)^*$
  is a regular property, since there are no constraints over the data. Now,
  checking that $\forall 1 \leq k \leq j, v_k = u_k$ requires a $j$-bounded
  counter (or no counter if $j = \infty$) which can be stored in memory, along
  with one register to store $\dt(u_k)$. Now, if $i = \infty$, there is no
  additional property to check. If $i < \infty$, checking that $\length{u'_1 \cdot
    u'_2 \dots \wedge v'_1 \cdot v'_2 \dots} \leq i \}$ requires two $i$-bounded
  counters to keep track of the respective lengths of the prefixes of $u' = u'_1
  \cdot u'_2 \dots$ and $v' = v'_1 \cdot v'_2 \dots$ which have been read so
  far, along with one register: the automaton guesses whether the mismatch
  position first appears in $u'$ or $v'$, stores the corresponding data in its
  additional register, as well as the label in its memory, goes to the
  corresponding position in $v'$ or $u'$ with the help of its $i$-bounded
  counters and checks that either the data or the label indeed mismatches.
\end{proof}

\section{Functionality, Equivalence and Composition of NRT}
\label{sec:functions}

In general, since they are non-deterministic, \nrt may not define functions but
relations, as illustrated by \autoref{ex:Trename}. In this section, we first
show that deciding whether a given \nrt defines a function is \PSpace-complete,
in which case we call it \emph{functional}. We show, as a consequence, that
testing whether two functional \nrt define two functions which coincide on their
common domain is \PSpace-complete. Finally, we show that functions defined by
\nrt are closed under composition. This is an appealing property in transducer
theory, as it allows to define complex functions by composing simple ones.
\begin{example}\label{ex2}
  As explained before, the transducer $T_\mathsf{rename}$ described in
  \autoref{ex:Trename} is not functional. To gain functionality, one can
  reinforce the specification by considering that one gets at the beginning a
  list of $k$ possible identifiers, and that one has to select the first one
  which is free, for some fixed $k$. This transformation is realised by the
  register transducer $T_\mathsf{rename2}$ depicted in \autoref{fig:Trename2}
  (for $k=2$).
\end{example}
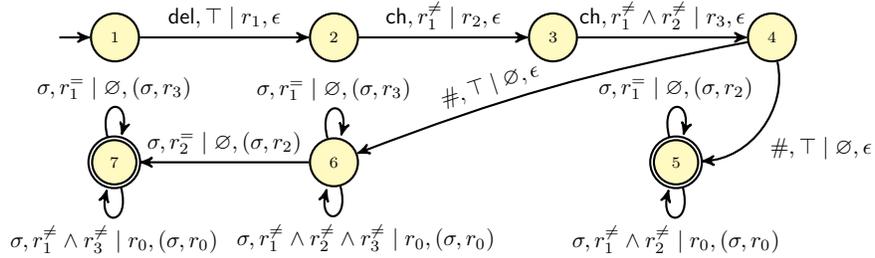
\begin{figure}[htb]
  \centering
  \begin{tikzpicture}[->,>=stealth',auto,node
    distance=2.25cm,thick,scale=0.9,every node/.style={scale=0.85}]
    \tikzstyle{every state}=[text=black, font=\scriptsize,
    fill=yellow!30,minimum size=7.5mm]

    \node[state, initial, initial text={}] (i) {1};
    \node[state, right=of i] (p) {2};
    \node[state, right=of p] (q) {3};
    \node[state, right=of q] (r) {4};
    \node[state, below=1cm of p] (u) {6};
    \node[state, left=of u,accepting] (v) {7};
    \node[state, accepting,xshift=-1.5cm] at (u-|r) (t) {5};

    \path (i) edge node {$\mathsf{del},\top \mid r_1,\epsilon $} (p);
    \path (p) edge node {$\mathsf{ch},r_1^{\neq} \mid r_2,\epsilon $} (q);
    \path (q) edge node {$\mathsf{ch},r_1^{\neq} \wedge r_2^{\neq} \mid r_3,\epsilon $} (r);
    \path (r) edge [bend left=50] node {$\#,\top \mid \varnothing,\epsilon$} (t);
    \path (r) edge[bend right=5] node[above,sloped,pos=0.65] {$\#,\top \mid \varnothing,\epsilon$} (u);
    \path (u) edge node[above] {$\sigma,r_2^{=} \mid \varnothing,(\sigma,r_2) $} (v);
    \path (t) edge[loop above] node {$\sigma,r_1^{=} \mid \varnothing,(\sigma,r_2) $} (t);
    \path (t) edge[loop below] node {$\sigma,r_1^{\neq} \wedge r_2^{\neq} \mid r_0,(\sigma, r_0) $} (t);
    \path (u) edge[loop above] node {$\sigma,r_1^{=} \mid \varnothing,(\sigma,r_3) $}
    (u);
    \path (u) edge[loop below] node[xshift=0.5cm]{$\sigma,r_1^{\neq} \wedge r_2^{\neq} \wedge r_3^{\neq} \mid r_0,(\sigma, r_0) $} (u);

    \path (v) edge[loop above] node {$\sigma,r_1^{=} \mid \varnothing,(\sigma,r_3) $}
    (v);
    \path (v) edge[loop below] node {$\sigma,r_1^{\neq} \wedge r_3^{\neq} \mid r_0,(\sigma, r_0) $} (v);
  \end{tikzpicture}
  \caption{A \nrt $T_\mathsf{rename2}$, with four registers $r_1, r_2, r_3$ and
    $r_0$ (the latter being used, as in \autoref{fig:Trename}, to output the last read data). After reading the $\#$ symbol, it guesses whether the value of
    register $r_2$ appears in the suffix of the input word. If not, it goes to
    state $5$, and replaces occurrences of $r_1$ by $r_2$. Otherwise, it moves
    to state $6$, waiting for an occurrence of $r_2$, and replaces occurrences
    of $r_1$ by $r_3$.}\label{fig:Trename2}
\end{figure}

\subsection{Functionality}

Let us start with the functionality problem in the data-free case. It is already
known that checking whether an \nft over $\omega$-words is functional is
decidable~\cite{DBLP:journals/iandc/CulikP81,DBLP:journals/ipl/Gire86}. By
relying on the pattern logic of~\cite{DBLP:conf/dlt/FiliotMR18} designed for
transducers of \emph{finite} words, we can show that it is decidable in
\NLogSpace.
\begin{proposition}
  \label{prop:funfinite}
  Deciding whether an \nft is functional is in \NLogSpace.
\end{proposition}
\begin{proof}
  Let $T$ be some \nft. By definition, $\rel{T}$ maps $\omega$-words to
  $\omega$-words, we have that $T$ is \emph{not} functional iff there exist
  three $\omega$-words $w,w_1,w_2$ such that $(w,w_1)\in \rel{T}$, $(w,w_2)\in
  \rel{T}$ and there is a mismatch between $w_1$ and $w_2$, i.e. a position $i$
  such that $w_1[i]\neq w_2[i]$. By taking a sufficiently long prefix $u$ of
  $w$, the latter is equivalent to the existence of finite runs
  $r_1:q_1\xrightarrow{u|v_1} p_1$ and $r_2:q_2\xrightarrow{u|v_2} p_2$, where
  $q_1,q_2$ are initial, $p_1,p_2$ are co-accessible by the same $\omega$-words
  (in the sense that there exist two accepting runs from $p_1$ and $p_2$ on the
  same $\omega$-word), and such that there is a mismatch between $v_1$ and
  $v_2$. This property is expressible in the pattern logic
  of~\cite{DBLP:conf/dlt/FiliotMR18} for transducers of finite words, whose
  model-checking problem is in \NLogSpace. More precisely, we model-check
  against $T$ the following pattern formula (using the syntax
  of~\cite{DBLP:conf/dlt/FiliotMR18}):
  \begin{align*}
    & \exists \pi_1:q_1\xrightarrow{u|v_1} p_1\exists \pi_2:q_2\xrightarrow{u|v_2}
      p_2 \\
    & \textsf{init}(q_1)\wedge \textsf{init}(q_2)\wedge
      \textsf{coacc}(p_1,p_2)\wedge \textsf{mismatch}(v_1,v_2)
  \end{align*}
  The predicate $\textsf{coacc}(p_1,p_2)$ is not directly defined
  in~\cite{DBLP:conf/dlt/FiliotMR18}, which is a logic for transducers over
  finite words, but we show that it is easily definable in the pattern logic.
  Indeed, the property of $(p_1,p_2)$ to be co-accessible by the same
  $\omega$-word is equivalent to asking the existence of finite runs $r'_1,r'_2$
  of the following form, where the outputs have not been indicated as they do
  not matter for this property:
  $$
  r'_1\ :\ p_1\xrightarrow{u_1} p'_1\xrightarrow{u_2} p''_1\xrightarrow{ u_3}
  p'_1\qquad r'_2\ :\ p_2\xrightarrow{u_1} p'_2\xrightarrow{u_2}
  p''_2\xrightarrow{ u_3} p'_2
  $$
  such that $p'_1,p''_2$ are accepting states. Indeed, if $(p_1,p_2)$ is
  co-accessible and since we use a B\"uchi accepting condition, it is
  co-accessible by a lasso word $u_1(u_2u_3)^\omega$ on which there is a run
  from $p_1$ with some accepting states $p'_1$ and a run from $p_2$ with some
  accepting state $p''_2$, both occurring on the loop of the lasso. Therefore,
  $\textsf{coacc}(p_1,p_2)$ holds true iff the following pattern logic formula
  is true against $T$ (seen as a transducer over finite words):
  \begin{align*}
    & \exists \pi'_1:p_1\xrightarrow{u_1}p'_1\exists 
      \pi''_1:p'_1\xrightarrow{u_2}p''_1 \exists 
      \pi'''_1:p''_1\xrightarrow{u_3}p'_1 \\
    & \exists \pi'_2:p_2\xrightarrow{u_1}p'_2\exists 
      \pi''_2:p'_2\xrightarrow{u_2}p''_2 \exists 
      \pi'''_2:p''_2\xrightarrow{u_3}p'_2 \\
    & \textsf{final}(p'_1)\wedge \textsf{final}(p''_2)
  \end{align*}
  
  This concludes the proof.
\end{proof}
The following theorem shows that a relation between data-words defined by an
\nrt with $k$ registers is a function iff its restriction to a set of data with
at most $2k+3$ data is a function. As a consequence, functionality is decidable
as it reduces to the functionality problem of transducers over a finite
alphabet.
\begin{theorem}
  \label{thm:functionality}
  Let $T$ be an \nrt with $k$ registers. Then, for all $X \subseteq \data$ of
  size $\size{X} \geq 2k+3$ such that $\dz \in X$, we have that $T$ is 
  functional if and only if $\rel{T} \cap ((\Sigma \times X)^\omega \times
  (\Gamma \times X)^\omega)$ is functional.
\end{theorem}
\begin{proof}
  The left-to-right direction is trivial. Now, assume $T$ is not functional. Let
  $x \in (\Sigma \times \data)^\omega$ be such that there exists $y,z \in
  (\Gamma \times \data)^\omega$ such that $y \neq z$ and $(x,y),(x,z) \in
  \rel{T}$. Let $i = \length{y \wedge z}$. Then, consider the language $L =
  \{\rho_1 \otimes \rho_2 \mid \rho_1$ and $\rho_2$ are accepting runs of $T,
  \projin(\rho_1) = \projin(\rho_2) \text{ and } \length{\projout(\rho_1) \wedge
    \projout(\rho_2)} \leq i\}$. Since, by \autoref{lem:otimes}, $L_\otimes(T)$
  is recognised by an \nra with $2k$ registers and, by \autoref{lem:mismatch},
  $M^i_{\infty}$
  is recognised by an \nra with $2$ registers, we get that $L = L_\otimes(T)
  \cap M^i_\infty$ is recognised by an \nra with $2k+2$ registers.

  Now, $L \neq \varnothing$, since, by letting $\rho_1$ and $\rho_2$ be the runs
  of $T$ both with input $x$ and with respective outputs $y$ and $z$, we have
  that $w = \rho_1 \otimes \rho_2 \in L$. Let $X \subseteq \data$ such that
  $\size{X} \geq 2k+3$ and $\dz \in X$. By
  \autoref{prop:indistinguishability}, we get that $L \cap (\Sigma
  \times X)^\omega \neq \varnothing$. By letting $w' = \rho'_1 \otimes \rho'_2
  \in L \cap (\Sigma \times X)^\omega$, and $x' = \projin(\rho'_1) =
  \projin(\rho'_2)$, $y' = \projout(\rho'_1)$ and $z' = \projout(\rho'_2)$, we
  have that $(x',y'), (x',z') \in \rel{T} \cap ((\Sigma \times X)^\omega \times (\Gamma \times X)^\omega)$ and $\length{y' \wedge z'} \leq i$,
  so, in particular, $y' \neq z'$ (since both are infinite words). Thus,
  $\rel{T} \cap ((\Sigma \times X)^\omega \times (\Gamma \times X)^\omega)$ is
  not functional.
\end{proof}

As a consequence of \autoref{prop:funfinite} and \autoref{thm:functionality}, we
obtain the following result. The lower bound is obtained by encoding
non-emptiness of register automata, which is
\PSpace-complete~\cite{DBLP:journals/tocl/DemriL09}.

\begin{corollary}
  \label{coro:fun}
  Deciding whether an \nrt $T$ is functional is \PSpace-complete.
\end{corollary}
\begin{proof}
  Let $T$ be an \nrt with $k$ registers. Choose any $X \subseteq \data$ of size
  $2k+3$ such that $\dz \in X$. By \autoref{thm:functionality}, $T$
  is functional if and only if $\rel{T} \cap ((\Sigma \times X)^\omega \times
  (\Gamma \times X)^\omega)$ is. By \autoref{prop:regularRestrictionTrd},
  $\rel{T} \cap (\Sigma \times X)^\omega \times(\Gamma \times X)^\omega$ is
  recognisable by an \nft with $O(\size{Q} \times \size{X}^{\size{R}})$ states,
  which can be constructed on-the-fly. The functionality problem for data-free
  transducers is in \NLogSpace by \autoref{prop:funfinite}, so this yields a
  \PSpace procedure.

  For the hardness, it is known that the emptiness problem of a register
  automaton $A$ is \PSpace-hard~\cite{DBLP:journals/tocl/DemriL09}. It can be
  easily reduced to a functionality problem, by constructing a transducer $T$
  which is functional iff its domain is empty, for instance a transducer
  realising the relation $f_1\cup f_2$ where $f_1 = \{ (wxw',wxw')\mid w\in
  L(A),w'\in(\Sigma\times \data)^\omega,x\in \{\#\}\times \data\}$ and $f_2 =
  \{(wxw',wx^\omega)\mid w\in L(A),w'\in(\Sigma\times \data)^\omega,x\in
  \{\#\}\times \data\}$, where $\#\not\in\Sigma$ is a new symbol. The transducer
  $T$ can be constructed from $A$ in polynomial time, and is functional iff
  $L(A) = \varnothing$.
\end{proof}

\subsection{Equivalence}
As a consequence of \autoref{coro:fun}, the following problem on the equivalence of \nrt is decidable:
\begin{theorem}
  The problem of deciding, given two functions $f,g$ defined by \nrt, whether
  for all $x\in\dom(f)\cap \dom(g)$, $f(x) = g(x)$, is \PSpace-complete.
\end{theorem}
\begin{proof}
  The formula $\forall x\in \dom(f)\cap \dom(g)\cdot f(x) = g(x)$ is true iff the
  relation $f\cup g = \{ (x,y)\mid y=f(x)\vee y=g(x)\}$ is a function. The
  latter can be decided by testing whether the disjoint union of the transducers
  defining $f$ and $g$ defines a function, which is in \PSpace by
  \autoref{coro:fun}. To show the hardness, we reduce the emptiness
  problem of \nra $A$ over finite words, just as in the proof of
  \autoref{coro:fun}. In particular, the functions $f_1$ and $f_2$ defined
  in this proof (which have the same domain) are equal iff
  $L(A)=\varnothing$.
\end{proof}

Note that under the promise that $f$ and $g$ have the same domain, then the
latter theorem implies that it is decidable to check whether the two functions
are equal. However, checking $\dom(f) = \dom(g)$ is undecidable, as the
language-equivalence problem for non-deterministic register automata is
undecidable, since, in particular, universality is
undecidable~\cite{Neven:2004:FSM:1013560.1013562} and the universal language
$(\Sigma \times \data)^\omega$ is trivially \nra-definable.

\subsection{Closure under composition}

Closure under composition is a desirable property for transducers, which 
holds in the data-free setting~\cite{berstel2009}. We show that it also
holds for functional \nrt.

\begin{theorem}\label{thm:closure}
  Let $f,g$ be two functions defined by \nrt. Then, their composition $f\circ g$
  is (effectively) definable by some \nrt.
\end{theorem}
\begin{proof}[Sketch]
  We first give a sketch of the proof. Then, the rest of this section is devoted
  to the details of the proof.

  By $f\circ g$ we mean $f\circ g: x \mapsto f(g(x))$. Assume $f$ and $g$ are
  defined by $T_f = (Q_f, R_f, q_0, F_f, \Delta_f)$ and $T_g = (Q_g, R_g, p_0,
  F_g, \Delta_g)$ respectively. Wlog we assume that the input and output finite
  alphabets of $T_f$ and $T_g$ are all equal to $\Sigma$, and that $R_f$ and
  $R_g$ are disjoint. We construct $T$ such that $\rel{T} = f\circ g$. The proof
  is similar to the data-free case where the composition is shown via a product
  construction which simulates both transducers in parallel, executing the second
  on the output of the first. Assume $T_g$ has some transition
  $p\xrightarrow{\sigma,\tst\mid \{r\},o} q$ where $o\in (\Sigma\times R_g)^*$.
  Then $T$ has to be able to execute transitions of $T_f$ while processing $o$,
  even though $o$ does not contain any concrete data values (it is here the main
  important difference with the data-free setting). However, if $T$ knows the
  equality types between $R_f$ and $R_g$, then it is able to trigger the
  transitions of $T_f$. For example, assume that $o = (a,r_g)$ and assume that
  the content of $r_g$ is equal to the content of $r_f$, $r_f$ being a register
  of $T_f$, then if $T_f$ has some transition of the form
  $p'\xrightarrow{a,r_f^{=}\mid \{r'_f\},o'} q'$ then $T$ can trigger the
  transition $(p,q)\xrightarrow{\sigma,\tst\mid \{r\}\cup \{r'_f := r_g\},o'}
  (p',q')$ where the operation $r'_f := r_g$ is a syntactic sugar on top of \nrt
  that intuitively means ``put the content of $r_g$ into $r'_f$''. In the
  following, we show that it is indeed syntactic sugar
  (\autoref{thm:reassign}) and provide more details on the construction.
\end{proof}

\subsubsection{NRT with reassignments}

An \emph{\nrt with reassignments} is defined as an \nrt where operations on
registers are of the form $r := curr$ to assign the current data read, or $r :=
r'$ to assign to $r$ the value of $r'$. It is required that each register $r$ appears
at most once as the left-hand-side of an assignment $r:=curr$ or $r:=r'$. Given
a configuration $(q,\tau)$ and some transition realising a set of instructions
$I$ and going to some state $q'$, the new register configuration $\tau'$ is
defined as follows. First, instructions $r := curr$ are executed, giving a new
register configuration $\tau''(s) = \tau(s)$ if $s$ has not been assigned the
current data, and $\tau''(s) =d$ otherwise, if $d$ is the current data. Then,
instructions $r:=r'$ are executed, so we let $\tau'(r) = \tau''(r')$ whenever
there is an instruction $r:=r'$, otherwise $\tau'(r)=\tau''(r)$.

We show in \autoref{thm:reassign} that this feature does not add expressive
power to the model.

\subsubsection{NRT with explicit tests}

First, to ease the construction, we need to explicit the tests, showing the
equivalence with register automata as defined
in~\cite{DBLP:conf/csl/Segoufin06}. Indeed, in our setting, they are represented
using logical formulas, which allows for a more compact representation, but
prevents us from keeping track of the equality relations between the different
registers, which will be needed in the construction. For instance, after a
transition $q\myxrightarrow{a, r_1 \vee r_2 \mid \{r_3\},r_1} q'$ is taken, we
do not know whether we have $r_1 = r_3$ or $r_2 = r_3$, since it depends whether
the current data is equal to the content of $r_1$ or of $r_2$. To explicit these
two cases, such transition can thus be split into two distinct transitions,
$q\myxrightarrow{a, r_1 \mid \{r_3\},r_1} q'$ and $q\myxrightarrow{a, r_2 \mid
  \{r_3\},r_1} q'$. Moreover, we also ignore whether such data is equal to some
content of some other register $r_4$, which could entail additional equalities
between registers. The two cases ($d = r_4$ and $d \neq r_4$) should thus yield
two different transitions. Such operations can be generalised, allowing for all
tests to be of the form $\tst_E = \bigwedge_{r \in E} r^= \wedge \bigwedge_{r
  \notin E} r^{\neq}$, where $E \subseteq R$ is a set of registers, meant to be
\emph{exactly} the set of registers to which the current data is equal. Thus,
when a data is read, it satifies at most one test, and the equality relations
between the registers can be updated.

Formally, a transition $q\xrightarrow{\sigma_\inp,\tst\mid
  \sigma_\outp,\asgn,r}_T q'$ is replaced by all the transitions
$q\xrightarrow{\sigma_\inp,\tst_E\mid \sigma_\outp,\asgn,r}_T q'$ for all
$E\subseteq R_k$ such that $\tst_E \Rightarrow \tst$ is true. Note that in
general, such operation adds exponentially many transitions (but does not affect
the number of states). In the rest of this section, we assume
that all transducers are in this normal form, and simply write $E$ for $\tst_E$.

Now, we can move to the removal of reassignments:
\begin{theorem}\label{thm:reassign}
  For all \nrt with reassignment $T$, one can construct an \nrt $T'$ such that
  $\rel{T} = \rel{T'}$.
\end{theorem}
\begin{proof}
  The idea is to keep in memory (in the states of $T'$) substitutions $\lambda$
  of registers by other registers, and define tests and outputs modulo those
  substitutions. We also add one extra register in $T'$. So, if $R$ is the set
  of registers of $T$, we let $R' = R\cup \{ s\}$ where $s$ is some extra
  register. The states of $T'$ maintain substitutions $\lambda$ of type
  $R\rightarrow R'$, where the initial substitution $\lambda_0$ is the identity
  on $R$. Hence, the set of states of $T'$ are pairs $(q,\lambda)$ where $q$ is a
  state of $T'$. Moreover, the reached substitution $\lambda$ 
  only depends on the
  sequence of transitions taken by $T$, so that we have a bijection between
  runs of $T$ and that of $T'$. The transducer $T'$ has the following invariant:
  for all runs of $T$ over some finite prefix $u$ reaching some configuration
  $(q,\tau)$, there exists a run of $T'$ on $u$ with the same sequence of
  $T$-states, reaching a configuration $((q,\lambda),\tau')$ such that for all
  $r\in R$, $\tau(r) = \tau'(\lambda(r))$. The converse also holds.

  Formally, $T'$ is constructed as follows: from any $T$-transition
  $q\myxrightarrow[T]{\sigma,E \mid \asgn,o} q'$ we create the
  $T'$-transitions 
  \begin{equation}
    \label{eq:derivedTransitions}
  (q,\lambda)\myxrightarrow[T']{\sigma,E'\mid \{r_0\},o'}
  (q',\lambda')
  \end{equation}
  for all substitutions $\lambda: R\rightarrow R'$, such that the
  following conditions hold:
  \begin{enumerate}

    \item $E' = \{ \lambda(r)\mid r\in E\}$
    \item $r_0\not\in \lambda(R)$ (it exists since
      $|R'|>|R|$). 
    \item  Let $\lambda''$ such that $\lambda''(r) = r_0$ if
      $(r:=curr)\in \asgn$, and $\lambda''(r) = \lambda(r)$
      otherwise. Then, for all $(r:=r')\in \asgn$, we let $\lambda'(r)
      = \lambda''(r')$, and if $r$ does not occur in some assignment
      $r:=r'$, then we let $\lambda'(r) = \lambda''(r)$. 
    \item if $o = (\sigma_1,r_1)\dots (\sigma_k,r_k)$, then $o' =
      (\sigma_1,\lambda'(r_1))\dots (\sigma_k,\lambda'(r_k))$. 
  \end{enumerate}

  Let us show that the invariant is preserved. It is obviously true for runs of
  length $0$. Now, let $r$ be a run of $T$ on a word
  $u\in(\Sigma\times \data)^*$, such that the last transition is 
  $(q,\tau)\myxrightarrow[T]{\sigma,E\mid \asgn,o}
    (q',\tau')$. By induction hypothesis there exists a run of $T'$
    on $u_1\dots u_{n-1}$ reaching a configuration
    $((q,\lambda),\gamma)$ such that $\tau = \gamma\circ
    \lambda$. We show that we can extend this run with a
    $T'$-transition from $(q,\lambda)$ as defined in \autoref{eq:derivedTransitions}, towards a
    configuration $((q',\lambda'),\gamma')$ such that $\tau' =
    \gamma'\circ \lambda'$. Let $d\in\data$ the last data
    of $u$. We know that $\tau(r) = d$ for all $r\in E$. We have to
    show that $\gamma(r') = d$ for all $r'\in E'$, i.e. $\gamma(\lambda(r)) =
    d$ for all $r\in E$, by definition of $E'$. This is immediate
    as $\tau = \gamma\circ \lambda$. Hence, the test is satisfied and
    the $T'$-transition defined in \autoref{eq:derivedTransitions} can be triggered.

    Now, let us show that $\tau' = \gamma'\circ \lambda'$. Let $r\in
    R$. First, note that $\gamma'(r) = \gamma(r)$ for all $r\neq
    r_0$. Then, let us consider several cases:
    \begin{itemize}
      \item if $r$ has been untouched by $\asgn$,
    then $\lambda'(r) = \lambda''(r) = \lambda(r)$ and $\tau'(r)
    = \tau(r)$, hence $\tau'(r) = \tau(r) = (\gamma\circ\lambda)(r) =
    (\gamma\circ \lambda')(r) = (\gamma'\circ\lambda')(r)$. The latter
    equality holds since
    $\gamma'(\lambda'(r)) = \gamma(\lambda'(r))$, because
    $\lambda'(r)\neq r_0$ (since $r$ has been untouched by $\asgn$). 
  \item if $(r:=curr)\in \asgn$, then $\lambda''(r) = r_0$ and since
    $(r:=r')\not\in \asgn$ for all $r'\in R$ (by our assumption that
    $r$ only occurs at most once in the lhs of assignments), we also
    have $\lambda'(r) = \lambda''(r) = r_0$. Hence $\gamma'\circ
    \lambda'(r) = \gamma'(r_0) = d$ and  $\tau'(r) = d$ since
    $(r:=curr)\in \asgn$, hence $\tau'(r) = \gamma'\circ \lambda'(r)$.  
  \item if $(r:=r')\in \asgn$, then $\lambda'(r) =
    \lambda''(r')$. Then there are again several cases:
    \begin{enumerate}
        \item $(r':=curr)\in \asgn$, then $\lambda''(r') = r_0$ and
          therefore $\gamma'\circ \lambda'(r) = \gamma'(r_0) =
          d$. Moreover, since $T$ first executes $r' := curr$ and then
          $r := r'$, we get $\tau'(r) = \tau(r) = d$, so that
          $\gamma'\circ \lambda'(r) = \tau'(r)$
        \item $(r':=curr)\not\in \asgn$, then $\lambda''(r') =
          \lambda(r')$. So, $\gamma'\circ \lambda(r) = \gamma'\circ
          \lambda(r')$. Necessarily,  $\lambda(r') \neq r_0$ by
          definition of $r_0$, which must satisfy $r_0\not\in
          \lambda(R)$. Therefore, $\lambda(r')$ is not assigned the
          current data value by $T'$, hence $\gamma'(\lambda(r')) =
          \gamma(\lambda(r'))$. Finally, we
          have:
          \begin{align*}
            \gamma'\circ \lambda' (r) & = \gamma'\circ \lambda''(r')
            \\
            & = \gamma'\circ \lambda(r') \\
            & = \gamma\circ \lambda(r') & \text{ since } \gamma'(\lambda(r')) =
          \gamma(\lambda(r')) \\
            & = \tau(r') & \text{by induction hypothesis} \\
            & = \tau'(r) & \text{since $(r:=r')\in\asgn$}
          \end{align*}
    \end{enumerate}
\end{itemize}
So, we have shown that the invariant is preserved. We have left to
show that $\gamma'(o') = \tau'(o)$, in the sense that if $r$ is the
$i$th register in $o$, then $\tau'(r) = \gamma'(r')$ where $r'$ is the
$i$th register of $o'$. By definition of $o'$, we have $r'=
\lambda'(r)$. By the invariant, $\tau'(r) = \gamma'\circ \lambda'(r)$,
hence $\tau'(r) = \gamma'(r')$. 

Overall, for all runs of $T$, there exists a run of $T'$
simulating it in the sense that it follows the same sequence of
$T$-states and produces the same output. This allows one to conclude
that $\rel{T}\subseteq \rel{T'}$. The other inclusion,
$\rel{T'}\subseteq \rel{T}$ also holds. It can be shown
similarly by induction on the length of runs. If $T'$ can trigger some
transition, this transition has been obtained from a transition $t$ of
$T$, and one can show using the same equalities and reasoning as above
that $t$ can also be triggered and that the same invariant as above is
satisfied, concluding the proof.
\end{proof}

\subsubsection{Proof of \autoref{thm:closure}, continued}
In the following construction, we first assume that $T_g$ produces at most one
letter at a time, i.e., the output $o\in (\Sigma\times \data)^*$ on all its
transitions is such that $|o| = 1$. We explain later on how to generalise it to
an arbitrary length. Under this assumption, we let $Q = Q_g\times Q_f\times
2^{R_g\times R_f}$ be the set of states of $T$ and $R = R_g\uplus R_f$ its set
of registers. Our construction satisfies the following invariant: if after
reading a prefix $x$, $T$ is in some configuration $((p,q,C),\tau)$, then $T_g$
is in state $p$, $T_f$ is in state $q$ after reading the output of $T_g$ on $x$,
and $(r_g,r_f)\in C$ iff $\tau(r_g) = \tau(r_f)$. The \emph{constraints} $C$ are
thus used to keep track of the equality relations between the registers of $T_f$
and $T_g$, a standard technique in the setting of register automata. We now
define the set of transitions. $(p,q,C)\xrightarrow{\sigma,E\mid \asgn,o}
(p',q',C')$ if the following conditions hold:
\begin{enumerate}
    \item there exists a transition
      $p\myxrightarrow[T_g]{\sigma,E_g\mid \asgn_g,(\sigma',r_g)} p'$,
    \item there exists a transition 
      $q\myxrightarrow[T_f]{\sigma',E_f\mid \asgn_f,o_f} q'$,
    \item let $C''$ be the intermediate equality type defined as the 
      equalities $(r_1,r_2)$ when $(r_1,r_2)$ were already equal and
      $r_1$ has not be reassigned the current data value, or $r_1$ has
      be reassigned the current data value, but this data value is
      equal to some $r'_1\in R_g$ (i.e. $r'_1\in E_g$) and $r'_1$
      was known to be equal to $r_2$. Formally:
      \begin{align*}
      C'' =  \{&(r_1,r_2)\mid (r_1,r_2)\in C\wedge
               r_1\not\in\asgn_g\} \\
        \cup\ \{&(r_1,r_2)\mid r_1\in\asgn_g\wedge
      \exists r'_1\in E_g\cdot (r'_1,r_2)\in C\}
      \end{align*}
      
      Then, we require that 
      for all $r_f\in R_f$, $r_f\in E_f$ iff $(r_g,r_f)\in C''$
    \item $\asgn = \asgn_g\cup \{ r := r_g\mid r\in \asgn_f\}$ (we
      assume that the instructions $r := r_g$ are executed after the
      assignments of the current data value)
    \item $o = o_f$
    \item $(r_1,r_2)\in C'$ iff  $(r_1,r_2)\in C''$ and $r_2\not\in
      \asgn_f$, or $r_2\in \asgn_f$ and $r_1=r_g$
    \end{enumerate}

    Finally, let us explain informally how to generalise this construction to
    outputs $o$ of arbitrarily lengths. Assume that $o =
    (\sigma_1,r_g^{1})(\sigma_2,r_g^{2})\dots (\sigma_n,r_g^{n})$.
    Note that $n$ is bounded and depends only on $T_f$. 
    The difference is at point $2$
    where instead we require the existence of a sequence of $n$
    transitions
    $$
    q_0\myxrightarrow[T_f]{\sigma_1,E_f^1\mid \asgn_f^1,o_f^1}
    q_1\myxrightarrow[T_f]{\sigma_2,E_f^2\mid \asgn_f^2,o_f^2}
    q_2\dots q_{n-1}\myxrightarrow[T_f]{\sigma_n,E_f^n\mid
      \asgn_f^n,o_f^n} q_n
    $$
    such that $q_0 = q$ and $q_n =q'$, and it satisfies some
    additional conditions that we now explain. First of all, starting
    from $C''$, since those transitions performs a series of
    assignments, the equality types between the registers of $T_g$ and
    $T_f$ may change along executing those transitions, which gives a
    series of equality types $C_1,\dots,C_n$ where $C_i$ is the
    equality type before executing the $i$th transition (with $C_1 =
    C''$ as defined above). Then, we
    require that for all $r_f\in R_f$, $r_f\in E_f^i$ iff $(r_g^i,
    r_f)\in C_i$. From $C_n$ and the last transition, we get a new
    equality type required to be equal to $C'$. We require that $\asgn
    = \asgn_g\cup \{ r := r_g^i\mid r\in \asgn_f^i$ and step $i$ is the
      last time $r$ is assigned in the sequence of
      transitions$\}$. Regarding the output, we cannot just require that $o
    = o_f^1\dots o_f^n$ as the values of the registers in $T_f$
    change along the sequence of transitions. However, for all $i$,
    the registers of $T_f$ are necessarily equal to some register of
    $T_g$, i.e. for all $r\in R_f$, there exists $r'\in R_g$ such that
    $(r',r)\in C_i$. This is due to the fact that each assignment
    $\asgn_i$ assigns the value of $r_g^i$ to the registers in
    $\asgn_i$  (assuming that the assignment of $T_f$ are always
    non-empty, which can be ensured wlog, modulo adding some register
    that always receives the current data value). Hence, in the end we
    require that $o = \alpha(C_1,o^1_f)\dots \alpha(C_n,o^n_f)$ where
    $\alpha(C_i,o^i_f)$ substitutes in $o^i_f$ any register $r_f\in
    R_f$ by some register $r\in R_g$ such that $(r,r_f)\in C_i$,
    concluding the proof.\qed

\begin{remark}
  The proof of \autoref{thm:closure} does not use the hypothesis that $f$
  and $g$ are functions, and actually shows a stronger result, namely that
  relations defined by \nrt are closed under composition.
\end{remark}

\section{Computability and Continuity}
\label{sec:computability-continuity}

We equip the set of (finite or infinite) data words with the usual distance: for
$u,v \in (\Sigma \times \data)^\omega$, $d(u,v) = 0$ if $u=v$ and $d(u,v) =
2^{-\length{u \wedge v}}$ otherwise. A sequence of (finite or infinite) data words
$(x_n)_{n\in\mathbb{N}}$ converges to some infinite data word $x$ if for all
$\epsilon>0$, there exists $N\ge0$ such that for all $n\ge N$, $d(x_n,x)\le
\epsilon$.

In order to reason with computability, we assume in the sequel that the infinite
set of data values $\data$ we are dealing with has an effective representation.
For instance, this is the case when $\data = \mathbb{N}$.

We now define how a Turing machine can compute a function of data
words. We consider deterministic Turing machines, which three tapes: a read-only
one-way input tape (containing the infinite input data word), a two-way working
tape, and a write-only one-way output tape (on which it writes the infinite output data word). Consider some input data word $x\in (\Sigma \times
\data)^\omega$. For any integer $k\in\mathbb{N}$, we let $M(x,k)$ denote the
output written by $M$ on its output tape after having read the $k$ first cells
of the input tape. Observe that as the output tape is write-only, the sequence
of data words $(M(x,k))_{k\ge 0}$ is non-decreasing.

\begin{definition}[Computability]
  A function $f:(\Sigma \times \data)^\omega \rightarrow (\Gamma \times
  \data)^\omega$ is computable if there exists a deterministic multi-tape
  machine $M$ such that for all $x\in\dom(f)$, the sequence $(M(x,k))_{k\ge 0}$
  converges to $f(x)$.
\end{definition}

\begin{definition}[Continuity]
  A function $f:(\Sigma \times \data)^\omega \rightarrow (\Gamma \times
  \data)^\omega$ is \emph{continuous} at $x \in \dom(f)$ if (equivalently):
  \begin{enumerate}[(a)]
  \item for all sequences of data words $(x_n)_{n \in \N}$ converging towards
    $x$, where for all $i \in \N$, $x_i \in \dom(f)$, we have that $(f(x_n))_{n
      \in \N}$ converges to $f(x)$.
  \item $\forall i \geq 0, \exists j \geq 0, \forall y \in \dom(f), \length{x
      \wedge y} \geq j \Rightarrow \length{f(x) \wedge f(y)} \geq i$.
  \end{enumerate}
  Then, $f$ is continuous if and only if it is continuous at each $x \in
  \dom(f)$. Finally, a functional \nrt $T$ is \emph{continuous} when $\rel{T}$
  is continuous.
\end{definition}

\begin{example}\label{ex:cont}
  We give an example of a non-continuous function $f$. The finite input and
  output alphabets are unary, and are therefore ignored in the description of
  $f$. Such function associates with every sequence $s =
  d_1d_2\dots\in\data^\omega$ the word $f(s) = d_1^\omega$ if $d_1$ occurs
  infinitely many times in $s$, otherwise $f(s) = s$ itself.

  The function $f$ is not continuous. Indeed, by taking $d \neq d'$, the
  sequence of data words $d(d')^nd^\omega$ converges to $d(d')^\omega$, while
  $f(d(d')^nd^\omega) = d^\omega$ converges to $d^\omega \neq f(d(d')^\omega) =
  d(d')^\omega$.

  Moreover, $f$ is realisable by some \nrt which non-deterministically guesses
  whether $d_1$ repeats infinitely many times or not. It needs only one register
  $r$ in which to store $d_1$. In the first case, it checks whether the current
  data $d$ is equal the content $r$ infinitely often, and in the second case, it
  checks that this test succeeds finitely many times, using B\"uchi conditions.

  One can show that the register transducer $T_\mathsf{rename2}$ considered in
  \autoref{ex2} also realises a function which is not continuous, as the
  value stored in register $r_2$ may appear arbitrarily far in the input word.
  One could modify the specification to obtain a continuous function as follows.
  Instead of considering an infinite log, one considers now an infinite sequence
  of finite logs, separated by $\$$ symbols. The register transducer
  $T_\mathsf{rename3}$, depicted in \autoref{fig:Trename3}, defines such a
  function.
\end{example}

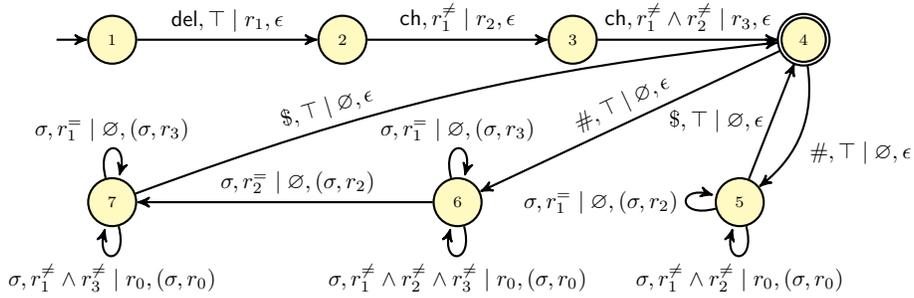
\begin{figure}[htb]
  \centering
  \begin{tikzpicture}[->,>=stealth',auto,node
    distance=2.4cm,thick,scale=0.9,every node/.style={scale=0.85}]
    \tikzstyle{every state}=[text=black, font=\scriptsize,
    fill=yellow!30,minimum size=7.5mm]

    \tikzstyle{input char}=[text=Red3] \tikzstyle{output char}=[text=Green4]

    \node[state, initial, initial text={}] (i) {1};
    \node[state, right=of i] (p) {2};
    \node[state, right=of p] (q) {3};
    \coordinate (pq) at ($(p)!0.5!(q)$);
    \node[state, right=of q,accepting] (r) {4};
    \node[state, below=1.5cm of i] (v) {7};
    \node[state] at (pq |- v) (u) {6};
    \node[state, xshift=-1cm] at (v-|r) (t) {5};

    \path (i) edge node {$\mathsf{del},\top \mid r_1,\epsilon $} (p);
    \path (p) edge node {$\mathsf{ch},r_1^{\neq} \mid r_2,\epsilon $} (q);
    \path (q) edge node {$\mathsf{ch},r_1^{\neq} \wedge r_2^{\neq} \mid r_3,\epsilon $} (r);
    \path (r) edge [bend left] node {$\#,\top\mid \varnothing,\epsilon $} (t);
    \path (r) edge node[above,sloped] {$\#,\top\mid \varnothing,\epsilon $} (u);
     \path (u) edge node[above,xshift=0.2cm] {$\sigma,r_2^{=} \mid
       \varnothing,(\sigma,r_2) $} (v);

     \path (t) edge[loop left] node {$\sigma,r_1^{=} \mid \varnothing,(\sigma,r_2) $}
     (t);
     \path (t) edge[loop below] node {$\sigma,r_1^{\neq} \wedge r_2^{\neq} \mid r_0,(\sigma, r_0) $} (t);
     \path (u) edge[loop above] node {$\sigma,r_1^{=}
       \mid \varnothing,(\sigma,r_3) $} (u);
     \path (u) edge[loop below] node {$\sigma,r_1^{\neq} \wedge r_2^{\neq} \wedge r_3^{\neq} \mid r_0,(\sigma, r_0) $} (u);
     \path (v) edge[loop above] node {$\sigma,r_1^{=} \mid \varnothing,(\sigma,r_3) $}
     (v);
     \path (v) edge[loop below] node {$\sigma, r_1^{\neq} \wedge r_3^{\neq} \mid r_0,(\sigma, r_0) $} (v);
     \path (v) edge [bend left=7] node[above,sloped,pos=0.3] {$\$,\top\mid
       \varnothing,\epsilon $} (r);
     \path (t) edge node[left] {$\$,\top\mid \varnothing,\epsilon $} (r);
   \end{tikzpicture}
   \caption{A register transducer $T_\mathsf{rename3}$. This transducer is
     non-deterministic, yet it defines a continuous
     function.}\label{fig:Trename3}
 \end{figure}

 We now prove the equivalence between continuity and computability for functions
 defined by \nrt. One direction, namely the fact that computability implies
 continuity, is easy, almost by definition. For the other direction, we rely on
 the following lemma which states that it is decidable whether a word $v$ can be
 safely output, only knowing a prefix $u$ of the input. In particular, given a
 function $f$, we let $\hat{f}$ be the function defined over all finite prefixes
 $u$ of words in $\dom(f)$ by $\hat{f}(u) = \bigwedge(f(uy)\mid uy\in\dom(f))$,
 the longest common prefix of all outputs of continuations of $u$ by $f$. Then,
 we have the following decidability result:
 \begin{lemma}
   \label{lem:updateout}
   The following problem is decidable. Given an $\nrt$ $T$ defining a function
   $f$, two finite data words $u\in (\Sigma\times \data)^*$ and $v\in
   (\Gamma\times \data)^*$, decide whether $v\preceq \hat{f}(u)$.
 \end{lemma}
\begin{proof}
  We decide the negation, i.e., whether there exists some continuation $y$ such
  that $uy\in\dom(f)$ and $v\not\preceq f(uy)$. Since $f(uy)$ is infinite, it is
  equivalent to asking whether there is a mismatch between $v$ and $f(uy)$ (a
  position $i$ such that $v[i]\neq f(uy)[i]$). We reduce this test to the
  language emptiness test of some \nra $A$. The language defined by $A$ is the
  set of data words $x_1y_1x_2y_2\dots$ such that:
  \begin{enumerate}
  \item for all $i\geq 1$, $x_i\in\Sigma\times \data$, i.e. $x_i$ is some input
  \item for all $j\geq 1$, $y_j\in (\Gamma\times \data)^*$, i.e. $y_j$ is some
    output data word
  \item $x_1\dots x_k = u$ for some $k\geq 1$
  \item $y_1\dots y_\ell$ mismatches with $v$ for some $\ell$
  \item there exists an accepting run of $T$ of the form:
    $$
    (q_0,\tau_0)\myxrightarrow[T]{x_1\mid y_1}
    (q_1,\tau_1)\myxrightarrow[T]{x_2\mid y_2} (q_2,\tau_2)\dots
    $$
  \end{enumerate}
  Note that $L(A)\neq \varnothing$ iff $v\not\preceq \hat{f}(u)$.

  It remains to show how to construct $A$. We construct $A$ by taking the
  intersection of three $\nra$ $A_3,A_4$ and $A_5$ which checks conditions
  $3,4,5$ ($\nra$ are closed under language intersection). To construct $A_3$,
  since $u$ is given, we can use as many registers as $|u|$ to store the data of
  $u$ in some initial register configuration, used to check condition $3$, by
  performing equality tests (the finite labels of $u$ can be dealt with by
  storing them in the states of $A_3$). This automaton needs $O(|u|)$ states. We
  can proceed similarly for $A_4$, with the difference that at some point, $A_4$
  has to guess the existence of some mismatch and perform some disequality test.
  Finally, $A_5$ just simulates $T$ over $x_1y_1\dots$.

  Since \nra emptiness is decidable~\cite{DBLP:journals/tocl/DemriL09} (even
  with some arbitrary initial configuration), this entails our result.
\end{proof}

 \begin{theorem} Let $f$ be a function defined by some \nrt $T$. Then $f$ is
   continuous iff $f$ is computable.
 \end{theorem}
 \begin{proof}
   $\Leftarrow$ Assuming $f=\rel{T}$ is computable by some Turing machine $M$,
   we show that $f$ is continuous. Indeed, consider some $x\in\dom(f)$, and some
   $i\ge 0$. As the sequence of finite words $(M(x,k))_{k\in\mathbb{N}}$
   converges to $f(x)$ and these words have non-decreasing lengths, there exists
   $j\ge0$ such that $|M(x,j)|\ge i$. Hence, for any data word $y\in\dom(f)$
   such that $|x\wedge y|\ge j$, the behaviour of $M$ on $y$ is the same during
   the first $j$ steps, as $M$ is deterministic, and thus $|f(x)\wedge f(y)|\ge
   i$, showing that $f$ is continuous at $x$.

   $\Rightarrow$ Assume that $f$ is continuous. We describe a Turing machine
   computing $f$; the corresponding algorithm is formalised as
   Algorithm~\ref{alg:computeContinuous}. When reading a finite prefix $x[{:}j]$
   of its input $x\in\dom(f)$, it computes the set $P_j$ of all configurations
   $(q,\tau)$ reached by $T$ on $x[{:}j]$. This set is updated along taking
   increasing values of $j$. It also keeps in memory the finite output word
   $o_j$ that has been output so far. For any $j$, if $\dt(x[{:}j])$ denotes the
   data that appear in $x$, the algorithm then decides, for each input
   $(\sigma,d)\in \Sigma\times (\dt(x[{:}j])\cup \{\dz\})$ whether $(\sigma,d)$
   can safely be output, i.e., whether all accepting runs on words of the form
   $x[{:}j]y$, for an infinite word $y$, outputs at least $o_j(\sigma,d)$. The
   latter can be decided, given $T$, $o_j$ and $x[{:}j]$, by
   \autoref{lem:updateout}. Note that it suffices to look at data in
   $\dt(x[{:}j])\cup \{\dz\}$ only since, by definition of \nrt, any data that
   is output is necessarily stored in some register, and therefore appears in
   $x[{:}j]$ or is equal to $\dz$.
   \begin{algorithm}[ht]
     \label{alg:computeContinuous}
     \SetAlgoLined
     \KwData{$x \in \dom(f)$} $o:=\epsilon$ \;
     \For{$j = 0$ \KwTo $\infty$}{
       \For{$(\sigma,d)\in \Sigma\times (\dt(x[{:}j])\cup \{\dz\})$}{
         \If(\tcp*[h]{such test is decidable by \autoref{lem:updateout}}){$o.(\sigma,d)\preceq \hat{f}(x[{:}j])$}{
           $o := o.(\sigma,d)$\; output $(\sigma,d)$\; } } }
     \caption{Algorithm describing the machine $M_f$ computing $f$.}
   \end{algorithm}

   Let us show that $M_f$ actually computes $f$. Let $x\in\dom(f)$. We have to
   show that the sequence $(M_f(x,j))_j$ converges to $f(x)$. Let $o_j$ be the
   content of variable $o$ of $M_f$ when exiting the inner loop at line 8, when
   the outer loop (line 2) has been executed $j$ times (hence $j$ input symbols
   have been read). Note that $o_j = M_f(x,j)$. We have $o_1\preceq o_2\preceq
   \dots$ and $o_j\preceq \hat{f}(x[{:}j])$ for all $j\geq 0$. Hence,
   $o_j\preceq f(x)$ for all $j\geq 0$. To show that $(o_j)_j$ converges to
   $f(x)$, it remains to show that $(o_j)_j$ is non-stabilising, i.e.
   $o_{i_1}\prec o_{i_2}\prec\dots$ for some infinite subsequence
   $i_1<i_2<\dots$. First, note that $f$ being continuous is equivalent to the
   sequence $(\hat{f}(x[{:}k]))_k$ converging to $f(x)$. Therefore we have that
   $f(x)\wedge \hat{f}(x[{:}k])$ can be arbitrarily long, for sufficiently large
   $k$. Let $j\geq 0$ and $(\sigma,d) = f(x)[|o_j|+1]$. By the latter property
   and the fact that $o_j.(\sigma,d)\preceq f(x)$, necessarily, there exists
   some $k>j$ such that $o_j.(\sigma,d)\preceq \hat{f}(x[{:}k])$. Moreover, by
   definition of \nrt, $d$ is necessarily a data that appears in some prefix of
   $x$, therefore there exists $k'\geq k$ such that $d$ appears in $x[{:}k']$
   and $o_j.(\sigma,d)\preceq \hat{f}(x[{:}k]\preceq \hat{f}(x[{:}k']$. This
   entails that $o_j.(\sigma,d)\preceq o_{k'}$. So, we have shown that for all
   for all $j$, there exists $k'>j$ such that $o_j\prec o_{k'}$, which concludes
   the proof.
 \end{proof}

 Now that we have shown that computability is equivalent with continuity for
 functions defined by \nrt, we exhibit a pattern which allows to decide
 continuity. Such pattern generalises the one
 of~\cite{DBLP:journals/corr/abs-1906-04199} to the setting of data words, the
 difficulty lying in showing that our pattern can be restricted to a finite
 number of data.
 \begin{theorem}
   \label{thm:characCont}
   Let $T$ be a functional \nrt with $k$ registers. Then, for all $X \subseteq
   \data$ such that $\size{X} \geq 2k+3$ and $\dz \in
   X$, 
   $T$ is not continuous at some $x \in (\Sigma \times \data)^\omega$ if and
   only if $T$ is not continuous at some $z \in (\Sigma \times X)^\omega$.
 \end{theorem}
 \begin{proof}
   The right-to-left direction is trivial. Now, let $T$ be a functional $\nrt$
   with $k$ registers which is not continuous at some $x \in (\Sigma \times
   \data)^\omega$. Let $f: \dom(\rel{T})
   \rightarrow (\Gamma \times \data)^\omega$ be the function defined by $T$,
   as: for all $u \in \dom(\rel{T}), f(u) = v$ where $v \in (\Gamma \times
   \data)^\omega$ is the unique data word such that $(u,v) \in \rel{T}$.

  
   Now, let $X \subseteq \data$ be such that $\size{X} \geq 2k+3$ and $\dz \in
   X$. We need to build two words $u$ and $v$ labelled over $X$ 
   which coincide on a sufficiently long prefix to allow for pumping, hence
   yielding a converging sequence of input data words whose images do not
   converge, witnessing non-continuity. To that end, we use a similar proof
   technique as for~\autoref{thm:functionality}: we show that the language of
   interleaved runs whose inputs coincide on a sufficiently long prefix while
   their respective outputs mismatch before a given position is recognisable by
   an \nra, allowing us to use the indistinguishability property. We also ask
   that one run presents sufficiently many occurrences of a final state $q_f$,
   so that we can ensure that there exists a pair of configurations containing
   $q_f$ which repeats in both runs.

   On reading such $u$ and $v$, the automaton behaves as a finite
   automaton, since the number of data is finite (\cite[Proposition
1]{Kaminski:1994:FA:194527.194534}). By analysing the respective runs,
   we can, using pumping arguments, bound the position on which the mismatch
   appears in $u$, then show the existence of a synchronised loop over $u$ and
   $v$ after such position, allowing us to build the sought witness for
   non-continuity.

   \myparagraph{Relabel over $X$} Thus, assume $T$ is not continuous at
   some point $x \in (\Sigma \times \data)^\omega$. Let $\rho$ be an accepting
   run of $T$ over $x$, and let $q_f \in \inf(\states(\rho)) \cap F$ be an
   accepting state repeating infinitely often in $\rho$. Then, let $i \geq 0$ be
   such that for all $j \geq 0$, there exists $y \in \dom(f)$ such that
   $\length{x \wedge y} \geq j$ but $\length{f(x) \wedge f(y)} \leq i$. Now,
   define $K = \size{Q} \times (2k+3)^{2k}$ and let $m = (2i+3) \times (K+1)$.
   Finally, pick $j$ such that $\rho[1{:}j]$ contains at least $m$ occurrences
   of $q_f$. Consider the language:
   \begin{align*}
     L = \bigl\{& \rho_1 \otimes \rho_2 \big\lvert \length{\projin(\rho_1) \wedge \projin(\rho_2)} \geq j, \length{\projout(\rho_1) \wedge \projout(\rho_2)} \leq i \text{ and } \\
                & \text{there are at least $m$
                  occurrences of $q_f$ in }\rho_1[1{:}j] \bigr\}
   \end{align*}
   By \autoref{lem:otimes}, $L_\otimes(T)$ is recognised by an \nra with $2k$
   registers. Additionnally, by \autoref{lem:mismatch}, $M_j^i$ is recognised
   by an \nra with $2$ registers. Thus, $L = L_\otimes(T) \cap O^{q_f}_{m,j}
   \cap M_j^i$, where $O^{q_f}_{m,j}$ checks there are at least $m$ occurrences
   of $q_f$ in $\rho_1[1{:}j]$ (this is easily doable from the automaton
   recognising $L_\otimes(T)$ by adding an $m$-bounded counter), is recognisable
   by an \nra with $2k+2$ registers.

   Choose $y \in \dom(f)$ such that $\length{x \wedge y} \geq j$ but
   $\length{f(x) \wedge f(y)} \leq i$. By letting $\rho_1$ (resp.
   $\rho_2$) be an accepting run of $T$ over $x$ (resp. $y$) we have $\rho_1 \otimes \rho_2 \in L$, so $L \neq \varnothing$. By
   \autoref{prop:indistinguishability}, it means $L \cap ((\Sigma
   \times X)^\omega \times (\Gamma \times X)^\omega) \neq \varnothing$. Let $w =
   \rho'_1 \otimes \rho'_2 \in L \cap ((\Sigma \times X)^\omega \times (\Gamma
   \times X)^\omega)$, $u = \projin(\rho'_1)$ and $v = \projin(\rho'_2)$. Then,
   $\length{u \wedge v} \geq j$, $\length{f(u) \wedge f(v)} \leq i$
   and there are at least $m$ occurrences of $q_f$ in $\rho_1[1{:}j]$.

   Now, we depict $\rho'_1$ and $\rho'_2$ in \autoref{fig:runsDiscontAsync},
   where we decompose $u$ as $u = u_1 \dots u_m \cdot s$ and $v$ as $v = u_1
   \dots u_m \cdot t$; their corresponding images being respectively $u' = u'_1
   \dots u'_{m} \cdot s'$ and $u'' = u''_1 \dots u''_{m} t''$. We also let $l =
   (i+1)(K+1)$ and $l' = 2(i+1)(K+1)$. Since the data of $u,v$ and $w$ belong to
   $X$, we know that $\tau_i, \mu_i : R \rightarrow X$.

  \begin{figure}[ht]
    \centering \resizebox{\textwidth}{!}{%
      \begin{tikzpicture}[->,>=stealth',auto,node
        distance=1.5cm,thick,scale=0.9,every node/.style={scale=0.85}]
        \tikzstyle{every state}=[text=black, fill=yellow!30, minimum size=1cm,
        initial text={}, font=\scriptsize]

        \node[state, initial] (q0) {$i_0,\dz^R$}; \node[right= of q0] (dot1)
        {$\dots$}; \node[state, right= of dot1,accepting] (qfk) {$q_f,\mu_{l}$};
        \draw [-,decorate,decoration={brace,amplitude=10pt,mirror,raise=0.3cm}]
        (q0) -- (qfk) node [black,midway,yshift=-1.3cm] {$(i+1)(K+1)$
          occurrences of $q_f$}; \node[right= of qfk] (dot2) {$\dots$};
        \node[state, right= of dot2,accepting] (qfl) {$q_f,\mu_{l'}$}; \draw
        [-,decorate,decoration={brace,amplitude=10pt,mirror,raise=0.3cm}] (qfk)
        -- (qfl) node [black,midway,yshift=-1.3cm] {$(i+1)(K+1)$ occurrences of
          $q_f$}; \node[right= of qfl] (dot3) {$\dots$}; \node[state, right= of
        dot3,accepting] (qfm) {$q_f,\mu_{m}$}; \draw
        [-,decorate,decoration={brace,amplitude=10pt,mirror,raise=0.3cm}] (qfl)
        -- (qfm) node [black,midway,yshift=-1.3cm] {$(K+1)$ occurrences of
          $q_f$};
        \node[right= of qfm] (tar1) {}; \node[state, initial,below=1.1cm of q0] (qp0)
        {$i_0,\dz^R$}; \node at (qp0-|dot1) (dot4) {$\dots$}; \node[state] at
        (qp0-|qfk) (pk) {$q_{l},\tau_{l}$}; \node at (qp0-|dot2) (dot5)
        {$\dots$}; \node[state] at (qp0-|qfl) (pl) {$q_{l'},\tau_{l'}$}; \node
        at (qp0-|dot3) (dot6) {$\dots$}; \node[state] at (qp0-|qfm) (pm)
        {$q_{m},\tau_{m}$};
        \node at (qp0-|tar1) (tar2) {};

        \path (q0) edge node {$u_1 \mid u'_1$} (dot1); \path (dot1) edge node
        {$u_{l} \mid u_{l}'$} (qfk); \path (qfk) edge node {$u_{l+1} \mid
          u_{l+1}'$} (dot2); \path (dot2) edge node {$u_{l'} \mid u_{l'}'$}
        (qfl); \path (qfl) edge node {$u_{l'+1} \mid u_{l'+1}'$} (dot3); \path
        (dot3) edge node {$u_m \mid u_m'$} (qfm); \path (qfm) edge node {$s|s'$}
        (tar1);

        \path (qp0) edge node {$u_1 \mid u''_1$} (dot4); \path (dot4) edge node
        {$u_{l} \mid u_{l}''$} (pk); \path (pk) edge node {$u_{l+1} \mid
          u_{l+1}''$} (dot5); \path (dot5) edge node {$u_{'} \mid u_{l'}''$}
        (pl); \path (pl) edge node {$u_{l'+1} \mid u_{l'+1}''$} (dot6); \path
        (dot6) edge node {$u_m \mid u_m''$} (pm);
        \path (pm) edge node {$t|t''$} (tar2);
      \end{tikzpicture}%
    }
    \caption{Runs of $f$ over $u = u_1 \dots u_m\cdot s$ and $v = u_1 \dots u_m
      \cdot t$.}
    \label{fig:runsDiscontAsync}
  \end{figure}
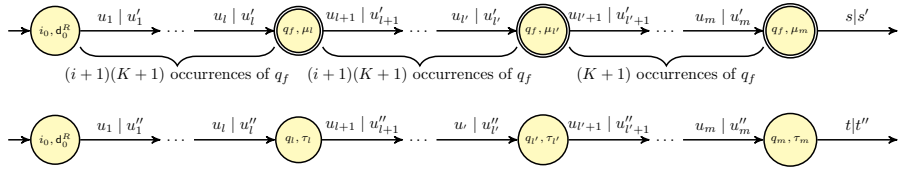

  \myparagraph{Repeating configurations} First, let us observe that in a partial
  run of $\rho'_1$ containing more than $\size{Q} \times \size{X}^k$ occurrences
  of $q_f$, there is at least one productive transition, i.e. a transition whose
  output is $o \neq \varepsilon$. Otherwise, by the pigeonhole principle, there
  exists a configuration $\mu : R \rightarrow X$ such that $(q_f, \mu)$ occurs
  at least twice in the partial run. Since all transitions are improductive, it
  would mean that, by writing $w$ the corresponding part of input, we have
  $(q_f, \mu) \myxrightarrow[T]{w \mid \varepsilon} (q_f, \mu)$. This partial
  run is part of $\rho'_1$, so, in particular, $(q_f, \mu)$ is accessible, hence
  by taking $w_0$ such that $(i_0,\tau_0) \myxrightarrow[T]{w_0 \mid w'_0}
  (q_f,\mu)$, we have that $f(w_0 w^\omega) = w'_0$, which is a finite word,
  contradicting our assumption that all accepting runs produce an infinite
  output. This implies that, for any $n \geq \size{Q} \times \size{X}^k$ (in
  particular for $n = l$), $\length{u'_1 \dots u'_n} \geq i+1$.

    \myparagraph{Locate the mismatch} Again, upon reading $u_{l+1} \dots u_{l'}$,
    there are $(i+1)(K+1)$ occurrences of $q_f$. There are two cases:
    \begin{enumerate}[(a)]
    \item There are at least $i+1$ productive transitions in $\rho'_2$. Then, we
      obtain that $\length{u''_1 \dots u''_{l'}} > i$, so
      $\mismatch(u'_1 \dots u'_{l'}, u''_1 \dots u''_{l'})$, since we know
      $\length{f(u) \wedge f(v)} \leq i$ and they are respectively prefixes of
      $f(u)$ and $f(v)$, both of length at least $i+1$.
      Afterwards, upon reading $u_{l'+1} \dots u_m$, there are $K+1 > \size{Q}
      \times \size{X}^{2k}$ occurrences of $q_f$, so, by the pigeonhole
      principle, there is a repeating pair: there exist indices $p$ and $p'$
      such that $l' \leq p < p' \leq m$ and $(q_f, \mu_p) = (q_f, \mu_{p'})$,
      $(q_p, \tau_p) = (q_{p'}, \tau_{p'})$. Thus, let $z_P = u_1 \dots u_{p}$,
      $z_R = u_{p+1} \dots u_{p'}$ and $z_C = u_{p'+1} \dots u_m \cdot t$ ($P$
      stands for \emph{prefix}, $R$ for \emph{repeat} and $C$ for
      \emph{continuation}; we use capital letters to avoid confusion with
      indices). By denoting $z'_P = u'_1 \dots u'_p$, $z'_R = u'_{p+1} \dots
      u'_{p'}$, $z''_P = u''_1 \dots u''_p$, $z''_R = u''_{p+1} \dots u''_{p'}$
      and $z''_C = u''_{p'+1} \dots u''_m \cdot t''$ the corresponding images,
      $z = z_P \cdot {z_R}^\omega$ is a point of discontinuity.
      Indeed, define $(z_n)_{n \in \N}$ as, for all $n \in \N$, $z_n = z_P \cdot
      z_R^n \cdot z_C$. Then, $(z_n)_{n \in \N}$ converges towards $z$,
      but, since for all $n \in \N$, $f(z_n)= z''_P \cdot {z''_L}^n \cdot
      z''_C$, we have that $f(z_n) \not
      \xrightarrow[n \infty]{} f(z) = z'_P \cdot {z'_L}^\omega$, since
      $\mismatch(z'_P,z''_P)$.
    \item Otherwise, by the same reasoning as above, it means there exists a
      repeating pair with only improductive transitions in between: there exist
      indices $p$ and $p'$ such that $l \leq p < p' \leq l'$, $(q_f, \mu_p) =
      (q_f, \mu_{p'})$, $(q_p, \tau_p) = (q_{p'}, \tau_{p'})$, and $(q_f, \mu_p)
      \myxrightarrow{u_{p+1} \dots u_{p'} \mid \varepsilon} (q_f, \mu_{p'})$,
      $(q_p,\tau_p) \myxrightarrow{u_{p+1} \dots u_{p'} \mid \varepsilon}
      (q_{p'}, \tau_{p'})$. Then, by taking $z_P = u_1 \dots u_{p}$, $z_R =
      u_{p+1} \dots u_{p'}$ and $z_C = u_{p'+1} \dots u_m \cdot t$, we have, by
      letting $z'_P = u'_1 \dots u'_p$, $z'_R = u'_{p+1} \dots u'_{p'}$, $z''_P
      = u''_1 \dots u''_p$, $z''_R = \varepsilon$ and $z''_C = u''_{n'+1} \dots
      u''_m \cdot t''$, that $z = z_P \cdot {z_R}^\omega$ is a point of
      discontinuity. Indeed, define $(z_n)_{n \in \N}$ as, for all $n \in \N$,
      $z_n = z_P \cdot z_R^n \cdot z_C$. Then, $(z_n)_{n \in \N}$ indeed
      converges towards $z$, but, since for all $n \in \N$, $f(z_n)= z''_P \cdot
      z''_C$, we have that $f(z_n) \not
      \xrightarrow[n \infty]{} f(z) = z'_P \cdot {z'_R}^\omega$, since
      $\mismatch(z'_P,z''_P \cdot z''_C)$ (the mismatch necessarily lies in
      $z'_P$, since $\length{z'_P} \geq i+1$). \qedhere
    \end{enumerate}
  \end{proof}

  \begin{corollary}
    \label{cor:continuity}
    Deciding whether an \nrt defines a continuous function is \linebreak
    \PSpace-complete.
  \end{corollary}
  \begin{proof}
    Let $X \subseteq \data$ be a set of size $2k+3$ containing $\dz$. By
    \autoref{thm:characCont}, $T$ is not continuous iff it is not continuous at
    some $z \in (\Sigma \times X)^\omega$, iff $\rel{T} \cap \bigl( (\Sigma
    \times X)^\omega \times (\Gamma \times X)^\omega \bigr)$ is not continuous.
    By \autoref{prop:regularRestrictionTrd}, such relation is recognisable by a
    finite transducer $T_X$ with $O(\size{Q} \times \size{X}^{\size{R}})$
    states, which can be built on-the-fly.
    By~\cite{DBLP:journals/corr/abs-1906-04199}, the continuity of functions
    defined by \nft is decidable in \NLogSpace, which yields a \PSpace
    procedure.

    For the hardness, we reduce again from the emptiness problem of register
    automata, which is \PSpace-complete~\cite{DBLP:journals/tocl/DemriL09}. Let
    $A$ be a register automaton over some alphabet $\Sigma\times \data$. We
    construct a transducer $T$ which defines a continuous function iff
    $L(A)=\varnothing$ iff the domain of $T$ is empty. Let $f$ be a
    non-continous function realised by some \nrt $H$ (it exists by
    \autoref{ex:cont}). Then, let $\#\not\in \Sigma$ be a fresh symbol, and
    define the function $g$ as the function mapping any data word of the form $w
    (\#,d) w'$ to $w (\#,d) f(w')$ if $w\in L(A)$. The function $g$ is realised
    by an \nrt which simulates $A$ and copies its inputs on the output to
    implement the identity, until it sees $\#$. If it was in some accepting
    state of $A$ before seeing $\#$, it branches to some initial state of $H$
    and proceeds executing $H$. If there is some $w_0\in L(A)$, then the
    subfunction $g_{w_0}$ mapping words of the form $w_0(\#,d)w'$ to $w_0
    (\#,d)f(w')$ is not continuous, since $f$ is not. Hence $g$ is not
    continuous. Conversely, if $L(A) = \varnothing$, then $\dom(g) =
    \varnothing$, so $g$ is continuous.
\end{proof}
In~\cite{DBLP:journals/corr/abs-1906-04199}, non-continuity is characterised by
a specific pattern (Lemma~21, Figure~1), i.e. the existence of some particular
sequence of transitions. By applying this characterisation to the finite
transducer recognising $\rel{T} \cap ((\Sigma \times X)^\omega \times (\Gamma
\times X)^\omega)$, as constructed in
\autoref{prop:regularRestrictionTrd}, we can characterise non-continuity
by a similar pattern, which will prove useful to decide (non-)continuity of
test-free \nrt in \NLogSpace (cf \autoref{sec:testFree}):
\begin{corollary}[\cite{DBLP:journals/corr/abs-1906-04199}]
  Let $T$ be an \nrt with $k$ registers. Then, for all $X \subseteq \data$ such
  that $\size{X} \geq 2k+3$ and $\dz \in X$, $T$ is not continuous at some $x
  \in (\Sigma \times \data)^\omega$ if and only if it has the pattern of
  \autoref{fig:patternCont}.

     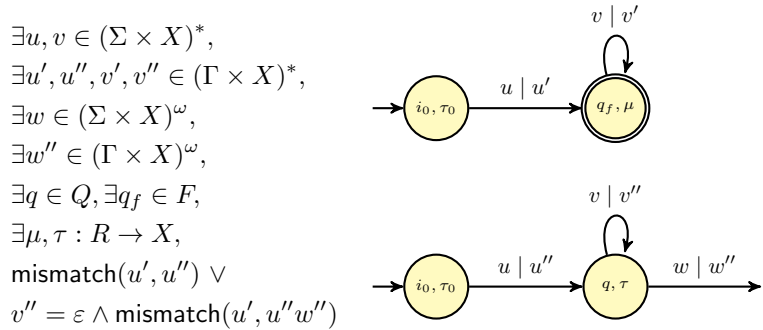
\begin{figure}[ht]
       \centering
       \begin{minipage}{0.4\textwidth}
         \begin{align*}
           & \exists u,v \in (\Sigma \times X)^*, \\
           & \exists u', u'', v', v'' \in (\Gamma \times X)^*, \\
           & \exists w \in (\Sigma \times X)^\omega, \\
           & \exists w'' \in (\Gamma \times X)^\omega, \\
           & \exists q \in Q, \exists q_f \in F, \\
           & \exists \mu, \tau : R \rightarrow X, \\
           &\mismatch(u',u'')~\vee \\
           &v'' = \varepsilon \wedge \mismatch(u',u''w'')
         \end{align*}
       \end{minipage}%
       \begin{minipage}{0.6\textwidth}
         \begin{tikzpicture}[->,>=stealth',auto,node
           distance=1.5cm,thick,scale=0.9,every node/.style={scale=0.85}]
           \tikzstyle{every state}=[text=black, fill=yellow!30,
           font=\scriptsize, minimum size=1cm]

           \node[state, initial, initial text={}] (p1) {$i_0,\tau_0$};
           \node[state, right= of p1, accepting] (q1) {$q_f,\mu$}; \node[state,
           initial, initial text={},below=of p1] (p2) {$i_0,\tau_0$};
           \node[state] at (p2 -| q1) (q2) {$q,\tau$}; \node[right= of q2] (c)
           {};

           \path (p1) edge node {$u \mid u'$} (q1); \path (p2) edge node {$u
             \mid u''$} (q2); \path (q1) edge[loop above] node {$v \mid v'$}
           (q1); \path (q2) edge[loop above] node {$v \mid v''$} (q2); \path
           (q2) edge node {$w \mid w''$} (c);
         \end{tikzpicture}
       \end{minipage}%
       \caption{A pattern characterising non-continuity of functions definable
         by an \nrt}
       \label{fig:patternCont}
     \end{figure}
   \end{corollary}

\section{Test-free Register Transducers}
\label{sec:testFree}
In~\cite{DBLP:conf/concur/ExibardFR19}, we introduced a restriction which allows
to recover decidability of the bounded synthesis problem for specifications
expressed as non-deterministic register automata. Applied to transducers,
such restriction also yields polynomial complexities when
considering the functionality and computability problems.

An \nrt $T$ is \emph{test-free} when its transition function does not depend on
the tests conducted over the input data. Formally, we say that $T$ is
\emph{test-free} if for all transitions $q \myxrightarrow[T]{\sigma, \tst \mid
    \asgn, o} q'$ we have $\tst = \top$.
Thus, we can omit the tests altogether and its transition relation can be
represented as $\Delta' \subseteq Q \times \Sigma \times 2^R \times (\Gamma
\times R)^* \times Q$.
\begin{example} Consider the function $f:(\Sigma \times \data)^\omega \rightarrow (\Gamma \times \data)^\omega$ associating, to $x = (\sigma_1, d_1) (\sigma_2, d_2)
  \dots$, the value $(\sigma_1, d_1) (\sigma_2, d_1) (\sigma_3, d_1) \dots$ if
  there are infinitely many $a$ in $x$, and $(\sigma_1, d_2) (\sigma_2, d_2)
  (\sigma_3, d_2) \dots$ otherwise.

  $f$ can be implemented using a test-free \nrt with one register: it initially
  guesses whether there are infinitely many $a$ in $x$, if it is the case, it
  stores $d_1$ in the single register $r$, otherwise it waits for the next input
  to get $d_2$ and stores it in $r$. Then, it outputs the content of $r$ along
  with each $\sigma_i$. $f$ is not continuous, as even outputting the first data
  requires reading an infinite prefix when $d_1 \neq d_2$.
\end{example}
Note that when a transducer is test-free, the existence of an accepting run over
a given input $x$ only depends on its finite labels. Hence, the
existence of two outputs $y$ and $z$ which mismatch over data can be
characterised by a simple pattern (\autoref{fig:patternMismatchTf}), which
allows to decide functionality in polynomial time:
\begin{theorem}
  \label{thm:functionalityTf}
  Deciding whether a test-free \nrt is functional is in \PTime.
\end{theorem}
\begin{proof}
  Let $T$ be a test-free \nrt such that $T$ is not functional. Then, there
  exists $x \in (\Sigma \times \data)^\omega$, $y,z \in (\Gamma \times
  \data)^\omega$ such that $(x,y), (x,z) \in \rel{T}$ and $y \neq z$. Then, let
  $i$ be such that $y[i] \neq z[i]$. There are two cases. Either $\lab(y[i])
  \neq \lab(z[i])$, which means that the finite transducer $T'$ obtained by
  ignoring the registers of $T$ is not functional. By \autoref{prop:funfinite},
  such property can be decided in \NLogSpace, so let us focus on the second
  case: $\dt(y[i]) \neq \dt(z[i])$.

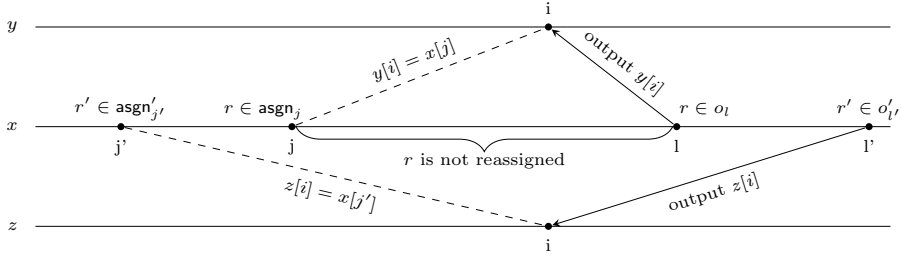
\begin{figure}[ht]
  \centering
\resizebox{\textwidth}{!}{%
\begin{tikzpicture}[x=1.2cm,y=0.7cm,font=\scriptsize]

  \tikzstyle{dot}=[circle,fill,inner sep=1pt]

  \draw (0,-2) -- (10,-2);
  \draw (0,0) -- (10,0);
  \draw (0,2) -- (10,2);

  \coordinate (refLab) at (0,0.35);
  \node[label=left:{$x$}] (p0) at (0,0) {};
  \node[label=left:{$y$}] (p0) at (0,2) {};
  \node[label=left:{$z$}] (p0) at (0,-2) {};
  \node[dot,label=below:{j}] (j1) at (3,0) {};
  \node[dot,label=below:{j'}] (j2) at (1,0) {};
  \node at (refLab-|j2) {$r' \in \asgn'_{j'}$};
  \node[xshift=-0.4cm,yshift=-0.03cm] at (refLab-|j1) {$r \in \asgn_j$};
  \node[dot,label=below:{l}] (l1) at (7.5,0) {};
  \node[xshift=0.4cm] at (refLab -| l1) {$r \in o_l$};
  \node[dot,label=below:{l'}] (l2) at (9.75,0) {};
  \node at (refLab -| l2) {$r' \in o'_{l'}$};
  \node[dot,label=above:{i}] (i1) at (6,2) {};
  \node[dot,label=below:{i}] (i2) at (6,-2) {};
  \draw [-,decorate,decoration={brace,amplitude=10pt,mirror}]
(j1) -- (l1) node [black,midway,below=0.25cm] {$r$ is not reassigned};

  \draw[dashed]       (j1) to node[midway,above,sloped]{$y[i] = x[j]$} (i1);
  \draw[dashed]       (j2) to node[midway,below,sloped]{$z[i] = x[j']$} (i2);
  \draw[->,>=stealth] (l1) to node[midway,above,sloped]{output $y[i]$} (i1);
  \draw[->,>=stealth] (l2) to node[midway,below,sloped]{output $z[i]$} (i2);
\end{tikzpicture}%
}
  \caption{A situation characterising the existence of a mismatch over data. Since
    acceptance does not depend on data, we can always choose $x$ such that
    $\dt(x[j]) \neq \dt(x[j'])$. Here, we assume that the labels of $x,y$ and
    $z$ range over a unary alphabet; in particular $y[i] = x[j]$ iff $\dt(y[i])
    = \dt(x[j])$. Finally, for readability, we did not write that $r'$ should
    not be reassigned between $j'$ and $l'$. Note that the position of $i$ with
    regards to $j,j',l$ and $l'$ does not matter; nor does the position of $l$
    w.r.t. $l'$.}
  \label{fig:patternMismatchTf}
\end{figure}

We first give an outline of the proof: observe that an input $x$ admits two
outputs which mismatch over data if and only if it admits two runs which
respectively store $x[j]$ and $x[j']$ such that $x[j] \neq x[j']$ and output
them later at the same output position $i$; the outputs $y$ and $z$ are then
such that $\dt(y[i]) \neq \dt(z[i])$. Since $T$ is test-free, the existence of
two runs over the same input $x$ only depends on its finite labels. Then, the
registers containing respectively $x[j]$ and $x[j']$ should not be reassigned
before being output, and should indeed output their content at the same position
$i$ (cf \autoref{fig:patternMismatchTf}). Besides, again because of
test-freeness, we can always assume that $x$ is such that $x[j] \neq x[j']$.
Overall, such pattern can be checked by a $2$-counter Parikh automaton, whose
emptiness is decidable in $\PTime$~\cite{DBLP:conf/lics/FigueiraL15} (under
conditions that are satisfied here).

Now, we move to the detailed proof. The case where the mismatch is over the labels reduces to deciding the
functionality of an \nrt, which is in \NLogSpace by \autoref{prop:funfinite}.
Let us thus move the case where the mismatch occurs over the data.

\myparagraph{Characterising a Mismatch}
Let us show that there exists an input $x \in
(\Sigma \times \data)^\omega$, two outputs $y,z \in (\Gamma \times
\data)^\omega$ such that $\dt(y) \neq \dt(z)$ if and only if there exists two
finite sequences of transitions $\nu : q_0 \myxrightarrow{\sigma_1 \mid \asgn_1,
  o_1} \dots \myxrightarrow{\sigma_n \mid \asgn_n, o_n} q_n$ and $\nu' : q'_0
\myxrightarrow{\sigma'_1 \mid \asgn'_1, o'_1} \dots \myxrightarrow{\sigma'_n
  \mid \asgn'_n, o'_n} {q'}_n$, two registers $r$ and $r'$, an (output) position
$i$ and (input) positions $j \neq j'$, $n \geq l \geq j$ and $n \geq l' \geq j'$
such that:
\begin{enumerate}
\item \label{itm:sameLabels} The input labels are the same: for all $1 \leq k
  \leq n, \sigma_k = \sigma'_k$
\item \label{itm:acceptingRuns} $\nu$ and $\nu'$ can yield accepting runs: $q_n$
  and $q'_n$ are coaccessible with the same data word
\item \label{itm:lInputOrigins} $l$ and $l'$ are the respective input positions
  at which the output data at position $i$ is produced: $\length{o_1 \dots
    o_{l-1}} < i \leq \length{o_1 \dots o_l}$ and $\length{o'_1 \dots o'_{l-1}} < i
  \leq \length{o'_1 \dots o'_l}$
\item \label{itm:rOutput} $r$ and $r'$ are the registers whose content is output
  at that moment: by denoting $w = o_1 \cdot o_2 \dots$ and $w' = o'_1 \cdot
  o'_2 \dots$, we have that $w[i] = (\sigma_i, r)$ and $w'[i] = (\sigma'_i, r')$
\item \label{itm:diffContent} $r$ and $r'$ are assigned at a different position,
  respectively $j$ and $j'$: $r \in \asgn_j$, $r' \in \asgn_{j'}$
\item \label{itm:rDataOrigins} $r$ and $r'$ are not reassigned before being
  output: $\forall j < k \leq l, r \notin \asgn_k$, $\forall j' < k \leq l', r
  \notin \asgn'_k$; in other words, $j$ and $j'$ are the respective
  \emph{origins} of the output at position $i$ over $\rho$ and $\rho'$, in the
  sense of~\cite{DBLP:conf/concur/ExibardFR19}
\end{enumerate}
  
$\Leftarrow$ Choose $x = (\sigma_1,d_1) \dots (\sigma_n, d_n) (\sigma_{n+1},
d_{n+1}) \dots$, such that $d_j \neq d_{j'}$. Then, since $T$ is test-free, by
items~\ref{itm:sameLabels} and~\ref{itm:acceptingRuns}, $\nu$ and $\nu'$ both
yield accepting runs of $T$, that we respectively denote $\rho$ and $\rho'$. We
then let $y = \projout(\rho)$ and $z = \projout(\rho')$. Now, let $(q_j,
\tau_j)$ be the configuration of $T$ at position $j$. Since $r \in \asgn_j$, we
have that $\tau_j(r) = d_j$. Then, by item~\ref{itm:rDataOrigins}, since for all
$j < k \leq l$, $r \notin \asgn_k$, we get that $\tau_l(r) = d_j$. Then, by
items~\ref{itm:lInputOrigins} and~\ref{itm:rOutput}, we get that $\dt(y[i]) =
d_j$. Similarly, we have that $\dt(z[i]) = d'_j$. Overall, we get that $x,y$ and
$z$ are such that $(x,y), (x,z) \in \rel{T}$ and $\dt(y) \neq \dt(z)$.

$\Rightarrow$ Now, let $x \in (\Sigma \times \data)^\omega$, $y, z \in (\Gamma
\times \data)^\omega$ be such that $(x,y), (x,z) \in \rel{T}$ and $\dt(y) \neq
\dt(z)$. Let $i \in \N$ be such that $\dt(y[i]) \neq \dt(z[i])$. Let $\rho$ and
$\rho'$ be two accepting runs of $T$ such that $\projin(\rho) = \projin(\rho') =
x$ and $\projout(\rho) = y$, $\projout(\rho') = z$. Then, let $l$ and $l'$ be
such that $\length{o_1 \dots o_{l-1}} < i \leq \length{o_1 \dots o_l}$ and
$\length{o'_1 \dots o'_{l-1}} < i \leq \length{o'_1 \dots o'_l}$, and define $n =
\max(l,l')$.
Let $\nu = q_0 \myxrightarrow{\sigma_1 \mid \asgn_1, o_1} q_1 \dots
\myxrightarrow{\sigma_n \mid \asgn_n, o_n} q_n$ and $\nu' = q'_0
\myxrightarrow{\sigma'_1 \mid \asgn'_1, o'_1} q'_1 \dots
\myxrightarrow{\sigma'_n \mid \asgn'_n, o'_n} q_n$ be their respective sequences
of transitions, truncated at length $n$.
  
First, for all $1 \leq k \leq n, \sigma_k = \sigma'_k = x[k]$, and $\sigma_{n+1}
\sigma_{n+2} \dots$ indeed yiels a final run from $q_n$ and $q'_n$ since $\rho$
and $\rho'$ are accepting.
Now, let $r$ and $r'$ be such that, by denoting $w = o_1 \cdot o_2 \dots$ and
$w' = o'_1 \cdot o'_2 \dots$, we have $w[i] = (\sigma_i, r)$ and $w'[i] =
(\sigma'_i, r')$. Then, define $j = \max \{k \leq l \mid r \in \asgn_k\}$ and
$j' = \max \{k' \leq l' \mid r' \in \asgn_{k'}\}$ (with the convention that
$\max \varnothing = 0$). By definition, item~\ref{itm:rDataOrigins} then holds
and $l \geq j$, $l' \geq j'$. Finally, by denoting $d_j = \dt(x[j])$ and $d_{j'}
= \dt(x[j'])$ (with the convention that $\dt(x[0]) = \dz$), we have that
$\tau_j(r) = d_j$ and $\tau_{j'}(r') = d_{j'}$. Since $r$ and $r'$ are not
reassigned before position $l$ and $l'$ respectively, we get that $\tau_l(r) =
d_j$ and $\tau_{l'}(r') = d_{j'}$, so $d_j = \dt(y[i])$ and $d_{j'} =
\dt(z[i])$, which means that $d_j \neq d_{j'}$, so $j \neq j'$.

\myparagraph{Recognising the Pattern}
Now, checking those properties can be done by a $2$-counter Parikh automaton. A
Parikh automaton is a finite automaton equipped with a finite number of counters
that it can increment and decrement, and which take their values in $\Z$ (i.e.
they are allowed to go below zero). There are no tests over the counters during
the run, but, to determine acceptance, at the end of the run, the automaton
checks whether its counters satisfy a given Presburger formula, or,
equivalently, belong to some semi-linear set. When increments and decrements are
encoded in unary, the semi-linear set is given explicitly, and the dimension
(i.e. the number of counters) is fixed, emptiness of Parikh automata is in
\PTime. This is the case here: there are only increments, which can be given in
unary as they correspond to the length of the output words in the transitions,
and at the end of the run, we only need to check that the two counters are
equal, i.e. belong to the diagonal $c_1=c_2$.

Let us describe the automaton: initially, it guesses the finite labels of $\nu$
and $\nu'$. Then, there are two phases: it first simulates the two partial runs
in parallel, only with regards to input labels. Formally, in this phase it can
take a transition $q \myxrightarrow{\sigma} q'$ if and only if there exists some
transition $q \myxrightarrow{\sigma \mid \asgn, o} q'$ in $T$. All the while, it
keeps track of the length of the output which has been produced by each run
using its two counters $c_1$ and $c_2$: $c_1$ is each time incremented by
$\length{o_1}$ (resp. $c_2$ by $\length{o_2}$). It then guesses position $j$ and
$j'$, and the corresponding registers $r$ and $r'$. Afterwards, it keeps
simulating both runs, but additionnally checks that $r$ (resp. $r'$) is not
reassigned after position $j$ (resp. $j'$). To that end, it avoids all
transitions $q \myxrightarrow{\sigma \mid \asgn, o} q'$ such that $\asgn \ni r$
(resp. $\asgn \ni r'$). Then, it guesses positions $l$ and $l'$ which are the
respective origins of the mismatch. At this point, it increments $c_1$ (resp.
$c_2$) by $\length{w}$, where $w$ is such that $o_l = w_1 r w_1'$ (resp.
$\length{w_2}$, where $w_2$ is such that $o_{l'} = w_2 r w'_2$) and then stops
counting the length of the corresponding outputs until the end of the run, i.e.
it stops incrementing $c_1$ (resp. $c_2$). Finally, at the end of the run, it
checks that $c_1 = c_2$, and that $q_n$ and $q'_n$ are co-accessible with the
same $\omega$-word (here, the registers do not matter anymore). This is doable
in polynomial time: precompute all pairs $(q_n, q'_n)$ such that
$\textsf{coacc}(q_n, q'_n)$ hold, such predicate being computable in \NLogSpace
as demonstrated in \autoref{prop:funfinite}.

Then, such automaton is polynomial in the size of $T$. Since emptiness of such
automaton is in $\PTime$, we can decide the above properties, i.e. whether $T$
is functional, in $\PTime$.
\end{proof}

Now, let us move to the case of continuity. Here again, the fact that test-free
\nrt conduct no test over the input data allows to focus on the only two
registers that are responsible for the mismatch, the existence of an accepting
run being only determined by finite labels.
\begin{theorem}
  \label{thm:continuityTf}
  Deciding whether a test-free \nrt defines a continuous function is in \PTime.
\end{theorem}
\begin{proof}
  Let $T$ be a test-free \nrt.
  \myparagraph{A simpler pattern}
  Let us first show that $T$ is continuous
  if and only if $T$ has the pattern of \autoref{fig:patternContTf},
  where $r$ is coaccessible (since acceptance only depends on finite labels, $T$
  can be trimmed\footnote{We say that $T$ is trim when all its states are both
    accessible and coaccessible.} in polynomial time).
\begin{figure}[ht]
  \centering
  \begin{minipage}{0.4\textwidth}
    \begin{align*}
      & \exists u,v \in (\Sigma \times \data)^*, \\
      & \exists u', u'', v', v'' \in (\Gamma \times \data)^*, \\
      & \exists z \in (\Sigma \times \data)^\omega, \\
      & \exists z'' \in (\Gamma \times \data)^\omega, \\
      & \exists q,r \in Q, \exists q_f \in F, \\
      &\mismatch(u',u'')~\vee \\
      &v'' = \varepsilon \wedge \mismatch(u',u''z'')
    \end{align*}
  \end{minipage}%
  \begin{minipage}{0.6\textwidth}
    \begin{tikzpicture}[->,>=stealth',auto,node
      distance=1.25cm,thick,scale=0.9,every node/.style={scale=0.85}]
      \tikzstyle{every state}=[text=black, font=\scriptsize,
      fill=yellow!30,minimum size=7.5mm]

      \node[state, initial, initial text={}] (p1) {$i_0$}; \node[state, right=
      of p1, accepting] (q1) {$q_f$}; \node[state, right= of q1, accepting] (r1)
      {$q_f$}; \node[state, initial, initial text={},below=of p1] (p2) {$i_0$};
      \node[state] at (p2 -| q1) (q2) {$q$}; \node[state] at (p2 -| r1) (r2)
      {$q$}; \node[state, right= of r2] (s) {$r$};

      \path (p1) edge node {$u \mid u'$} (q1); \path (p2) edge node {$u \mid
        u''$} (q2); \path (q1) edge node {$v \mid v'$} (r1); \path (q2) edge
      node {$v \mid v''$} (r2); \path (r2) edge node {$z \mid z''$} (s);
    \end{tikzpicture}
  \end{minipage}%
  \caption{A pattern characterising non-continuity of
    functions definable by an \nrt. Configurations are not depicted, as there
    are no conditions on them. We unrolled the loops to highlight the fact that
    they do not necessarily loop back to the same configuration.}
  \label{fig:patternContTf}
  \end{figure}
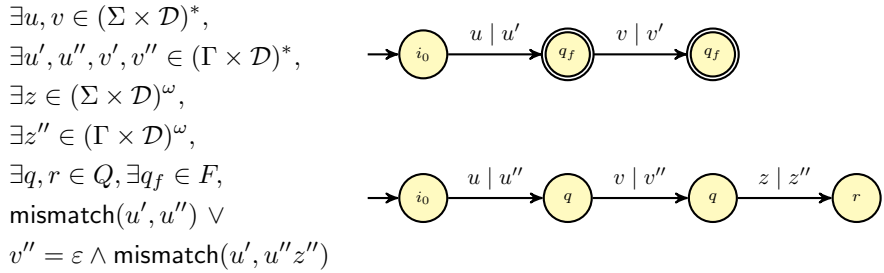
  By \autoref{thm:characCont}, we already know that $T$ is not continuous if and only if it has the
  pattern of \autoref{fig:patternCont}. Now, the left-to-right direction is quite direct, since the pattern of \autoref{fig:patternContTf} is simpler: assume $T$ is
  not continuous. Then it has the pattern of \autoref{fig:patternCont}. Thus, if
  $\mismatch(u',u'')$, by picking $w = w'' = \varepsilon$ and $r = q$, we get
  the pattern of \autoref{fig:patternContTf}. Otherwise, choose two finite data
  words $z$ and $z''$ and $r$ such that $z \prec w$, $z'' \prec w''$, $q
  \myxrightarrow{z \mid z''} q'$ and $\mismatch(u',u''z'')$. Such finite words
  exist because the mismatch between $u'$ and $u'' w''$ happens at some finite
  position, and it suffices to truncate $w$ and $w''$ up to such position.

  Conversely, assume $T$ has the pattern of \autoref{fig:patternContTf}. Then, let $x =
  uv^{\omega}$, and define $(y_n)_{n \in \N}$ as $y_n = uv^nwz$ for all $n \in
  \N$, where $z \in (\Sigma \times \omega)$ is such that there is a final run
  over $z$ from $r$. Such $z$ exists since $r$ is coaccessible; in the following
  we denote $z''$ its corresponding image. Then, $y_n \xrightarrow[n \rightarrow
  \infty]{} x$; however, for all $n \in \N$, $f(y_n) = u'' (v'')^n w'' z''$, so
  $\length{f(x) \wedge f(y_n)} \leq \length{u'}$, since either $v'' \neq
  \varepsilon$ and then there is a mismatch between $u'$ and $u''$, or there is
  a mismatch between $u'$ and $u''w''$, which means that $T$ is not continuous.

  \myparagraph{Checking the pattern}
  Now, it remains to show that such simpler pattern can be checked in $\PTime$. We treat
  each part of the disjunction separately:
  \begin{enumerate}[(a)]
  \item there exists $u,u',u'',v,v',v''$ such that $i_0 \myxrightarrow{u \mid
      u'} q_f \myxrightarrow{v \mid v'} q_f$ and $i_0 \myxrightarrow{u \mid u''}
    q \myxrightarrow{v \mid v''} q$, where $q_f \in F$ and $\mismatch(u',u'')$.
    Then, as shown in the proof of \autoref{thm:functionalityTf}, there exists a
    mismatch between some $u'$ and $u''$ produced by the same input $u$ if and
    only if there exists two runs and two registers $r$ and $r'$ assigned at two
    distinct positions, and later on output at the same position.
    Such pattern can be checked by a $2$-counter Parikh automaton similar to the
    one described in the proof of \autoref{thm:functionalityTf}; the only
    difference is that here, instead of checking that the two end states are
    coaccessible with a common $\omega$-word, we only need to
    check that $q_f \in F$ and that there is a synchronised loop over $q_f$ and
    $q$, which are regular properties that can be checked by the Parikh
    automaton with only a polynomial increase.
  \item there exists $u,u',u'',v,v',z,z''$ such that $i_0 \myxrightarrow{u \mid
      u'} q_f \myxrightarrow{v \mid v'} q_f$ and $i_0 \myxrightarrow{u \mid u''}
    q \myxrightarrow{v \mid \varepsilon} q \myxrightarrow{z \mid z''} r$, where
    $q_f \in F$ and $\mismatch(u',u''z'')$. By examining again the proof of
    \autoref{thm:functionalityTf}, it can be shown that to obtain a mismatch, it
    suffices that the input is the same for both runs only up to position
    $\max(j,j')$. More precisely, there is a mismatch between $u'$ and $u''z''$
    if and only if there exists two registers $r$ and $r'$ and two positions
    $j,j' \in \{1, \dots, \length{u}\}$ such that $j \neq j'$, $r$ is stored at
    position $j$, $r'$ is stored at position $j'$, $r$ and $r'$ are respectively
    output at input positions $l \in \{1, \dots, \length{u}\}$ and $l' \in \{1,
    \dots, \length{uz}\}$ and they are not reassigned in the meantime. Again, such
    property, along with the fact that $q_f \in F$ and the existence of a
    synchronised loop can be checked by a $2$-counter Parikh automaton of
    polynomial size.
  \end{enumerate}
  Overall, deciding whether a test-free \nrt is continuous is in \PTime.
\end{proof}


%
%
%
\bibliographystyle{splncs04} \bibliography{Bibliography}
%


\vfill

{\small\medskip\noindent{\bf Open Access} This chapter is licensed under the terms of the Creative Commons\break Attribution 4.0 International License (\url{http://creativecommons.org/licenses/by/4.0/}), which permits use, sharing, adaptation, distribution and reproduction in any medium or format, as long as you give appropriate credit to the original author(s) and the source, provide a link to the Creative Commons license and indicate if changes were made.}

{\small \spaceskip .28em plus .1em minus .1em The images or other third party material in this chapter are included in the chapter's Creative Commons license, unless indicated otherwise in a credit line to the material.~If material is not included in the chapter's Creative Commons license and your intended\break use is not permitted by statutory regulation or exceeds the permitted use, you will need to obtain permission directly from the copyright holder.}

\medskip\noindent\includegraphics{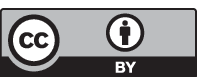}

\end{document}